\documentclass[runningheads]{llncs}
\title{Computing Weighted Subset Transversals in $H$-Free Graphs\thanks{The research in this paper received support from the Leverhulme Trust (RPG-2016-258).
The first author was also supported by a Rutherford Foundation Postdoctoral Fellowship, administered by the Royal Society Te Ap$\bar{\mbox{a}}$rangi.
  An extended abstract of this paper will appear in the proceedings of WADS 2021~\cite{BJP21}.}}

\author{Nick Brettell\inst{1}\orcidID{0000-0002-1136-418X} 
\and
Matthew Johnson\inst{2}\orcidID{0000-0002-7295-2663}
\and \\
Dani\"el Paulusma\inst{2}\orcidID{0000-0001-5945-9287}}

\authorrunning{N. Brettell, M. Johnson, D. Paulusma}

\institute{School of Mathematics and Statistics, Victoria University of Wellington, \\ New Zealand, \email{nick.brettell@vuw.ac.nz} \and
Department of Computer Science, Durham University, UK, \\ \email{\{matthew.johnson2,daniel.paulusma\}@durham.ac.uk}}

\newcounter{ctrclaim}[theorem]
\newcounter{ctrcase}[theorem]
\newcounter{ctrstep}[theorem]
\newcounter{ctrsubstep}[ctrstep]
\newcounter{ctrsubsubstep}[ctrsubstep]

\newcommand\displaycase[1]{{\bfseries#1}}

\newcommand{\thmstep}[1]{\medskip\phantomsection\refstepcounter{ctrstep}\noindent\displaycase{Step \thectrstep. }{\itshape #1}\\}

\newcommand\faketheorem[1]{{\bfseries#1}}

\usepackage[T1]{fontenc}
\usepackage{comment,amsmath,amssymb,xspace,xcolor,tikz,graphicx,tabularx,tablefootnote,boxedminipage}
\usetikzlibrary{fit,automata,arrows}
\usetikzlibrary{shapes}
\usetikzlibrary{decorations.pathreplacing}
\usetikzlibrary{decorations.pathmorphing}

\usepackage{cleveref}

\newcommand{\optproblemdef}[3]{
	\begin{center}
		\begin{boxedminipage}{.99\textwidth}
			\textsc{{#1}}\\[2pt]
			\begin{tabular}{ r p{0.8\textwidth}}
				\textit{~~~~Instance:} & {#2}\\
				\textit{Task:} & {#3}
			\end{tabular}
		\end{boxedminipage}
	\end{center}
}

\newcommand{\problemdef}[3]{
	\begin{center}
		\begin{boxedminipage}{.99\textwidth}
			\textsc{{#1}}\\[2pt]
			\begin{tabular}{ r p{0.8\textwidth}}
				\textit{~~~~Instance:} & {#2}\\
				\textit{Question:} & {#3}
			\end{tabular}
		\end{boxedminipage}
	\end{center}
}

\newcommand \dia{\hfill{$\diamond$}}

\newcommand{\NP}{{\sf NP}}

\newcommand{\ssi}{\subseteq_i}
\newcommand{\si}{\supseteq_i}

\begin{document}

\maketitle

\begin{abstract} 
For the {\sc Odd Cycle Transversal} problem, the task is to find a small set $S$ of vertices in a graph that intersects every cycle of odd length.
The {\sc Subset Odd Cycle Transversal} problem requires $S$ to intersect only those odd cycles that include a vertex of a distinguished vertex subset $T$.
If we are given weights for the vertices, we  ask instead that~$S$ has small weight: this is the problem {\sc Weighted Subset Odd Cycle Transversal}.
We prove an almost-complete complexity dichotomy for {\sc Weighted Subset Odd Cycle Transversal} for graphs that do not contain a graph~$H$ as an induced subgraph.  
Our general approach can also be used for
{\sc Weighted Subset Feedback Vertex Set}, which enables us to generalize a recent result of Papadopoulos and Tzimas.

\end{abstract}

\section{Introduction}\label{s-intro}

For a \emph{transversal} problem, one seeks to find a small set of vertices within a given graph that intersects every subgraph of a specified kind. 
Two problems of this type are {\sc Feedback Vertex Set} and {\sc Odd Cycle Transversal}, where the objective is to find a 
small set~$S$ of vertices that intersects, respectively, every cycle and every cycle with an odd number of vertices. 
Equivalently, when $S$ is deleted from the graph, what remains is  
a forest or a bipartite graph, respectively.

For a \emph{subset transversal} problem, we are also given a vertex subset~$T$ and we must find a small set of vertices that intersects every subgraph of a specified kind \emph{that also contains a vertex of $T$}.  
An {\it (odd) $T$-cycle} is a cycle of the graph (with an odd number of vertices) that intersects~$T$. 
A set $S_T\subseteq V$~is a {\it $T$-feedback vertex set} or an {\it odd $T$-cycle transversal} of a graph $G=(V,E)$ if $S_T$ has at least one vertex of, respectively, every $T$-cycle or every odd $T$-cycle; see also Fig.~\ref{f-example}.
A \emph{(non-negative) weighting} of $G$ is a function 
$w: V\to \mathbb{R}^+$.
For $v \in V$, $w(v)$ is the \emph{weight} of $v$, and for $S \subseteq V$, the weight $w(S)$ of $S$ is the sum of the weights of the vertices in~$S$.   
In a  \emph{weighted subset transversal} problem the task is to find a transversal whose weight is less than a prescribed bound.
We study the following problems:

\problemdef{{\sc Weighted Subset Feedback Vertex Set}}
{a graph $G$, a subset $T\subseteq V(G)$, a non-negative vertex weighting~$w$ of $G$ and an integer $k\geq 1$.}
{does $G$ have a $T$-feedback vertex set $S_T$ with $w(S_T)\leq k$?} 

\problemdef{{\sc Weighted Subset Odd Cycle Transversal}}
{a graph $G$, a subset $T\subseteq V(G)$, a non-negative vertex weighting~$w$ of $G$ and an integer $k\geq 1$.}
{does $G$ have an odd $T$-cycle transversal $S_T$ with $w(S_T)\leq k$?} 

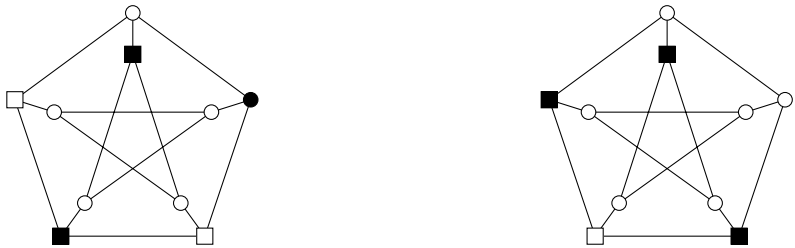
\begin{figure}[b]
\begin{center}
\begin{minipage}{0.45\textwidth}
\centering
\begin{tikzpicture}[xscale=0.55, yscale=0.55]
\draw (0,2)--(1.16,-1.6)--(-1.9,0.6)--(1.9,0.6)--(-1.16,-1.6)--(0,2) (-2.85,0.9)--(-1.74,-2.4)--(1.74,-2.4)--(2.85,0.9)--(0,3)--(-2.85,0.9) 
(-1.9,0.6)--(-2.85,0.9) (-1.16,-1.6)--(-1.74,-2.4) (1.16,-1.6)--(1.74,-2.4) (1.9,0.6)--(2.85,0.9) (0,2)--(0,3);
\draw[fill=white](-1.9,0.6) circle [radius=5pt] (-1.16,-1.6) circle [radius=5pt] (1.16,-1.6) circle [radius=5pt] (1.9,0.6) circle [radius=5pt] 
(0,2) node[regular polygon,regular polygon sides=4,draw,fill=black,scale=0.7pt] {} (-2.85,0.9) node[regular polygon,regular polygon sides=4,draw,fill=white,scale=0.7pt] {} 
(-1.74,-2.4) node[regular polygon,regular polygon sides=4,draw,fill=black,scale=0.7pt] {} (1.74,-2.4) node[regular polygon,regular polygon sides=4,draw,fill=white,scale=0.7pt] {} (0,3) circle [radius=5pt];
\draw[fill=black](2.85,0.9) circle [radius=5pt];
\end{tikzpicture}
\end{minipage}
\qquad
\begin{minipage}{0.45\textwidth}
\centering
\begin{tikzpicture}[xscale=0.55, yscale=0.55]
\draw (0,2)--(1.16,-1.6)--(-1.9,0.6)--(1.9,0.6)--(-1.16,-1.6)--(0,2) (-2.85,0.9)--(-1.74,-2.4)--(1.74,-2.4)--(2.85,0.9)--(0,3)--(-2.85,0.9)
(-1.9,0.6)--(-2.85,0.9) (-1.16,-1.6)--(-1.74,-2.4) (1.16,-1.6)--(1.74,-2.4) (1.9,0.6)--(2.85,0.9) (0,2)--(0,3);
\draw[fill=white] (-1.9,0.6) circle [radius=5pt] (-1.16,-1.6) circle [radius=5pt] (1.16,-1.6) circle [radius=5pt] (1.9,0.6) circle [radius=5pt]
(2.85,0.9) circle [radius=5pt] (0,2) node[regular polygon,regular polygon sides=4,draw,fill=black,scale=0.7pt] {}
(-2.85,0.9) node[regular polygon,regular polygon sides=4,draw,fill=black,scale=0.7pt] {} (-1.74,-2.4) node[regular polygon,regular polygon sides=4,draw,fill=white,scale=0.7pt] {}
(1.74,-2.4) node[regular polygon,regular polygon sides=4,draw,fill=black,scale=0.7pt] {}(0,3) circle [radius=5pt];
\end{tikzpicture}
\end{minipage}
\caption{Two examples (from~\cite{BJPP20}) of the Petersen graph with the set $T$ indicated by square vertices. The set $S_T$ of black vertices forms both an odd $T$-cycle transversal and a $T$-feedback vertex set. On the left, $S_T\cap T\neq \emptyset$. On the right, $S_T\subseteq T$.}\label{f-example}.
\end{center}
\end{figure}

\noindent
Both problems are \NP-complete 
even when the weighting function is~$1$ and \mbox{$T=V$}.
We continue a systematic study of transversal problems on hereditary graph classes, focusing on the weighted subset variants.  \emph{Hereditary} graph classes can be characterized by a set of forbidden induced subgraphs.
We begin with the case where this set has size~$1$:  the class of graphs that, for some graph~$H$, do not contain~$H$ as an induced subgraph; a graph in this class is said to be {\it $H$-free}.

\smallskip
\noindent
{\bf Past Results.}
We first note some \NP-completeness results for the special case where $w\equiv 1$ and $T=V$, which corresponds to the original problems
{\sc Feedback Vertex Set} and {\sc Odd Cycle Transversal}.  These results immediately imply \NP-completeness for the weighted subset problems.

By Poljak's construction~\cite{Po74}, for every integer~$g\geq 3$, {\sc Feedback Vertex Set}  is \NP-complete for graphs of finite girth at least~$g$ (the girth of a graph is the length of its shortest cycle).  There is an analogous result for {\sc Odd Cycle Transversal}~\cite{CHJMP18}.  
It has also been shown that {\sc Feedback Vertex Set}~\cite{Mu17b} and {\sc Odd Cycle Transversal}~\cite{CHJMP18} are \NP-complete for line graphs and, therefore, also for claw-free graphs.  Thus the two problems are \NP-complete for the class of $H$-free graphs whenever $H$ contains a cycle or claw.  Of course, a graph with no cycle is a forest, and a forest with no claw has no vertex of degree at least 3.  Hence, we need now only focus on the case where $H$ is a {\it linear forest}, that is, a collection of disjoint paths. 

There is no linear forest $H$ for which {\sc Feedback Vertex Set} on $H$-free graphs is known to be \NP-complete, but 
for {\sc Odd Cycle Transversal} we can take $H=P_2+P_5$ or $H=P_6$, as the latter problem
 is \NP-complete even for $(P_2+P_5,P_6)$-free graphs~\cite{DFJPPP19}. It is known that 
 {\sc Subset Feedback Vertex Set}~\cite{FHKPV14} 
 and  {\sc Subset Odd Cycle Transversal}~\cite{BJPP20}, which are the special cases with $w\equiv 1$, are
\NP-complete for $2P_2$-free graphs; in fact, these results were proved for split graphs which form a proper subclass of $2P_2$-free graphs.
For the weighted subset problems, there is just one additional case of \NP-completeness currently known, from the interesting recent work of 
Papadopoulos and Tzimas~\cite{PT20} as part of the following dichotomy. 

\begin{theorem}[\cite{PT20}]\label{t-known}
{\sc Weighted Subset Feedback Vertex Set} on $sP_1$-free graphs is polynomial-time solvable if 
$s\leq 4$ and \NP-complete if $s\geq 5$.
\end{theorem}

\noindent
The unweighted version of {\sc Subset Feedback Vertex Set} can be solved in polynomial time for $sP_1$-free graphs for every $s\geq 1$~\cite{PT20}.
In contrast, for many transversal problems, the complexities on the weighted and unweighted versions for $H$-free graphs align; see, for example {\sc Vertex Cover}~\cite{GKPP19}, {\sc Connected Vertex Cover}~\cite{JPP20} and {\sc (Independent) Dominating Set}~\cite{LMMZ20}. 
Thus {\sc Subset Feedback Vertex Set} is one of the few known problems for which, on certain hereditary graph classes, the (unweighted) problem is polynomial-time solvable, but the weighted variant is $\NP$-complete.

The other known polynomial-time algorithm for {\sc Weighted Subset Feedback Vertex Set} on $H$-free graphs is for the case where $H=P_4$.
This can be proven in two ways: {\sc Weighted Subset Feedback Vertex Set} is polynomial-time solvable for permutation graphs~\cite{PT19}
and also for graphs for which we can find a decomposition of constant mim-width in polynomial time~\cite{BPT19}; both classes contain the class of $P_4$-free graphs. To the best of our knowledge, algorithms for {\sc Weighted Subset Odd Cycle Transversal} on $H$-free graphs have not previously been studied.

We now mention the polynomial-time results on $H$-free graphs for the unweighted subset variants of the problems (which do not imply anything for the weighted subset versions). It is known that {\sc Subset Odd Cycle Transversal}~\cite{BJPP20} is polynomial-time solvable on $P_4$-free graphs~\cite{BJPP20}, and in Section~\ref{s-rm} we show that this result can be generalized to the weighted variant in a straightforward way. It is also known that {\sc Subset Feedback Vertex Set} and {\sc Subset Odd Cycle Transversal} are polynomial-time solvable for $(sP_1+P_3)$-free graphs for every integer $s\geq 0$~\cite{BJPP20}.   

Additionally, {\sc Weighted Feedback Vertex Set} is polynomial-time solvable on $P_5$-free graphs~\cite{ACPRS20} 
and $sP_3$-free graphs for every integer~$s\geq 1$~\cite{PPR21}. 
Moreover, {\sc Odd Cycle Transversal} is polynomial-time solvable on $sP_2$-free graphs for every $s\geq 1$~\cite{CHJMP18}. 
The latter result can be directly generalized to the weighted variant (see Appendix~\ref{a-sp2}).
Finally, {\sc Weighted Odd Cycle Transversal} is polynomial-time solvable for $(sP_1+P_3)$-free graphs; this follows from a straightforward adaptation of the proof for the unweighted variant given in~\cite{DFJPPP19} (see Appendix~\ref{a-b}).\footnote{The result for {\sc Odd Cycle Transversal} on $(sP_1+P_3)$-free graphs ($s\geq 1$) from~\cite{DFJPPP19} was shown before
the corresponding result for {\sc Subset Odd Cycle Transversal} was proven in~\cite{BJPP20}.}

\medskip
\noindent
{\bf Our Results.}
We enhance the current understanding of the two weighted subset transversal problems, presenting new polynomial-time algorithms for {\sc Weighted Subset Odd Cycle Transversal} and {\sc Weighted Subset Feedback Vertex Set} on $H$-free graphs for certain $H$. We highlight that {\sc Subset Odd Cycle Transversal} is a problem whose weighted variant is harder than its unweighted variant.
Our main result is the following almost-complete dichotomy. We write $H\ssi G$, or $G \si H$ 
to say that $H$ is an induced subgraph of $G$ (that is, $H$ can be obtained from $G$ by a sequence of vertex deletions).

\begin{theorem}\label{t-realmain}
Let $H$ be a graph with $H\notin \{2P_1+P_3,P_1+P_4,2P_1+P_4\}$. 
Then {\sc Weighted Subset Odd Cycle Transversal} on $H$-free graphs is polynomial-time solvable if 
$H\ssi 3P_1+P_2$, $P_1+P_3$, or $P_4$, and is \NP-complete otherwise.
\end{theorem}

\noindent
As a consequence, we obtain a dichotomy analogous to Theorem~\ref{t-known}.

\begin{corollary}
The {\sc Weighted Subset Odd Cycle Transversal} problem on $sP_1$-free graphs is polynomial-time solvable if 
$s\leq 4$ and is \NP-complete if $s\geq 5$.
\end{corollary}

\noindent
For the hardness part of Theorem~\ref{t-realmain} it suffices to show hardness for 
$H=5P_1$; this follows from the same reduction used by Papadopoulos and Tzimas~\cite{PT20} to prove 
 Theorem~\ref{t-known}. The three tractable cases, where $H\in \{P_4,P_1+P_3,3P_1+P_2\}$, are all new. Out of these cases, $H=3P_1+P_2$ is the most involved. 
 For this case we use a different technique to that used in~\cite{PT20}.
Although we also reduce to the problem of finding a minimum weight vertex cut that separates two given terminals, our technique relies less on explicit distance-based arguments, and we devise a method for distinguishing cycles according to parity. 
 Our technique also enables us to extend the result of~\cite{PT20} on {\sc Weighted Subset Feedback Vertex Set} from $4P_1$-free graphs to $(3P_1+P_2)$-free graphs, leading to the same almost-complete dichotomy for {\sc Weighted Subset Feedback Vertex Set}.

\begin{theorem}\label{t-realmain2}
Let $H$ be a graph with $H\notin \{2P_1+P_3,P_1+P_4,2P_1+P_4\}$. 
Then {\sc Weighted Subset Feedback Vertex Set} on $H$-free graphs is polynomial-time solvable if 
$H\ssi 3P_1+P_2$, $P_1+P_3$, or $P_4$, and is \NP-complete otherwise.
\end{theorem}
 
\noindent
We refer to Table~\ref{t-table} for an overview of the current knowledge of the problems, including the results of this paper. 

\newcommand\Tstrut{\rule{0pt}{2.6ex}}         
\newcommand\Bstrut{\rule[-0.9ex]{0pt}{0pt}}   

\begin{table}[t]
\begin{tabularx}{\textwidth}{ p{1.5cm} p{3.25cm} X  p{2.6cm}  }
    \hline
    & polynomial-time 
    & unresolved & $\NP$-complete\Tstrut\Bstrut \\
  \hline
    {\sc (W)FVS} &  $H \ssi P_5$ or \newline \phantom{$H \ssi$} $sP_3$ for $s\geq 1$  & $H\si P_1+P_4$ & none\Tstrut\Bstrut \\  \hline
   {\sc (W)OCT} & $H=P_4$ or \newline  $H\ssi sP_1+P_3$ or \newline \phantom{$H\ssi$} $sP_2$ for $s\geq 1$ &    
    $H=sP_1+P_5$ for $s\geq 0$  or \newline $H=sP_1+tP_2+uP_3+vP_4$ \newline for $s,t,u \ge 0$, $v\geq 1$
    \newline
    with $\min\{s,t,u\}\geq 1$ if $v=1$, or
    \newline
    $H=sP_1+tP_2+uP_3$ for $s,t\geq 0$, $u\geq 1$
    with $u\geq 2$ if $t=0$
     & $H\si P_6$ or $P_2+P_5$\Tstrut\Bstrut \\  \hline
{\sc SFVS}, \newline
{\sc SOCT} & $H=P_4$ or \newline $H\ssi sP_1+P_3$ for $s\geq 1$ & $H = sP_1+P_4$ for $s\geq 1$ & $H \si 2P_2$\Tstrut\Bstrut \\  \hline
{\sc WS\nolinebreak FVS}, \newline
{\sc WSOCT} & $H \ssi {P_4},\textcolor{blue}{P_1+P_3}$, or \newline \phantom{$H\ssi$} $\textcolor{blue}{3P_1+P_2}$  & $H\in \{2P_1+P_3,P_1+P_4, 2P_1+P_4\}$ & $H \si  5P_1$ or $2P_2$\Tstrut\Bstrut \\
   
\hline  \\[-0.25cm]
\end{tabularx}
 \caption{
 The complexity of {\sc (Weighted) Feedback Vertex Set} ((W)FVS), {\sc (Weighted) Odd Cycle Transversal} ((W)OCT), and their subset (S) and weighted subset (WS) variants, when restricted to $H$-free graphs for linear forests $H$.  All problems are $\NP$-complete for $H$-free graphs when $H$ is not a linear forest. The four blue cases (two for WSFVS, two for WSOCT) are the {\it algorithmic} contributions of this paper; see also Theorems~\ref{t-realmain} and~\ref{t-realmain2}.}
\label{t-table}
\end{table}

\section{Preliminaries}

Let $G=(V,E)$ be a graph. If $S\subseteq V$, then~$G[S]$ denotes the subgraph~of~$G$ induced by~$S$, and~$G-S$ is the graph $G[V\setminus S]$. The path on $r$ vertices is denoted $P_r$.
We say that~$S$ is {\it independent} if $G[S]$ has no edges, and that $S$ is a {\it clique} and $G[S]$ is {\it complete} if every pair of vertices in $S$ is joined by an edge.

If~$G_1$ and~$G_2$ are vertex-disjoint graphs, then
the \emph{union} operation $+$ creates the disjoint union $G_1+\nobreak G_2$ having vertex set $V(G_1)\cup V(G_2)$ and edge set $E(G_1)\cup E(G_2)$.  By $sG$, we denote the disjoint union of $s$ copies of $G$.
Thus $sP_1$ denotes the graph whose vertices form an independent set of size~$s$.
 
A {\it  (connected) component} of~$G$ is a maximal connected subgraph of $G$.
The graph $\overline{G}=(V,\{uv\; |\; uv\not \in E\; \mbox{and}\; u\neq v\})$ is the \emph{complement} of~$G$.
The \emph{neighbourhood} of a vertex $u\in V$ is the set $N_G(u)=\{v\; |\; uv\in E\}$. For $U\subseteq V$, we let $N_G(U)=\bigcup_{u\in U}N(u)\setminus U$.

Let $S$ and $T$ be two disjoint vertex sets of a graph $G$.
Then~$S$ is \emph{complete} to~$T$ if every vertex of~$S$ is adjacent to every vertex of~$T$, and~$S$ is \emph{anti-complete} to~$T$ if there are no edges between~$S$ and~$T$. 
In the first case, $S$ is also said to be \emph{complete} to~$G[T]$, and in the second case we say it is \emph{anti-complete} to~$G[T]$.

A graph is \emph{bipartite} if its vertex set can be partitioned into at most two independent sets.
A  graph is \emph{complete bipartite} if its vertex set can be partitioned into two independent sets~$X$ and~$Y$ such that $X$ is complete to~$Y$. If $X$ or $Y$ has size~$1$, the complete bipartite graph is a {\it star}; recall that $K_{1,3}$ is also called a claw.

\section{General Framework of the Algorithms}\label{s-general}

We first explain our general approach with respect to odd cycle transversals. Afterwards we modify our terminology for feedback vertex sets, but we note that our approach can be easily extended to other kinds of transversals as well. 

So, consider an instance $(G,T,w)$ of {\sc Weighted Subset Odd Cycle Transversal}.
Recall that a cycle is a {\it $T$-cycle} if it contains a vertex of $T$.
A subgraph of $G$ with no odd $T$-cycles is {\it $T$-bipartite}. Note that a subset $S_T\subseteq V$ is an odd $T$-cycle transversal if and only if $G[V\setminus S_T]$ is $T$-bipartite.
A {\it solution} for $(G,T,w)$ is an odd $T$-cycle transversal $S_T$. From now on,
whenever $S_T$ is defined,
we  let $B_T = V(G) \setminus S_T$ denote the vertex set of the corresponding $T$-bipartite graph.
If $u\in B_T$ belongs to at least one odd cycle of $G[B_T]$, then $u$ is an {\it odd} vertex of $B_T$. Otherwise, when $u\in B_T$ is not in any odd cycle of $G[B_T]$, we say that $u$ is an 
{\it even} vertex of $B_T$. Note that by definition every vertex in $T\cap B_T$ is even.
We let $O(B_T)$ and $R(B_T)$ denote the sets of odd and even vertices of $B_T$ (so $B_T=O(B_T)\cup R(B_T)$).

A solution $S_T$ is {\it neutral} if $B_T$ consists of only even vertices; in this case $S_T$ is an odd cycle transversal of $G$. We say that
$S_T$ is {\it $T$-full} if $B_T$ contains no vertex of $T$.
If $S_T$ is neither neutral nor $T$-full, then $S_T$ is a {\it mixed} solution.
We can now outline our approach 
to finding minimum weight odd $T$-cycle transversals:\\[-17pt]
\begin{enumerate}
\item Compute a neutral solution of minimum weight.
\item Compute a $T$-full solution of minimum weight.
\item Compute a mixed solution of minimum weight.
\item From the three computed solutions, take one of overall minimum weight.\\[-17pt]
\end{enumerate}
\noindent
As mentioned, a neutral solution is a minimum-weight odd cycle transversal. Hence, in Step~1, we will use existing polynomial-time algorithms from the literature for computing such an odd cycle transversal (these algorithms must be for the weighted variant). Step~2 is trivial: we can just set $S_T:=T$ (as $w$ is non-negative). Hence, most of our attention will go to Step~3.
For Step~3, we analyse the structure of the graphs $G[R(B_T)]$ and $G[O(B_T)]$ for a mixed solution $S_T$ and how these graphs relate  to each other.

For  {\sc Weighted Subset Feedback Vertex Set} we follow exactly the same approach, but we use slightly different terminology. 
A  subgraph of a graph $G=(V,E)$ is a {\it $T$-forest} if it has no $T$-cycles.  Note that a subset $S_T\subseteq V$ is a $T$-feedback vertex set if and only if $G[V\setminus S_T]$ is a $T$-forest. We write $F_T=V\setminus S_T$ in this case. 
If $u\in F_T$ belongs to at least one cycle of $G[F_T]$, then $u$ is a {\it cycle vertex} of $F_T$. Otherwise, if $u\in F_T$ is not in any cycle of $G[F_T]$, we say that $u$ is a {\it forest vertex} of $F_T$. By definition every vertex in $T\cap F_T$ is a forest vertex.

We obtain our results for {\sc Weighted Subset Feedback Vertex Set} by a simplification of our algorithms for {\sc Weighted Odd Cycle Transversal}. Hence, to explain our approach fully, we will now give a polynomial-time algorithm for 
{\sc Weighted Odd Cycle Transversal} for $(3P_1+P_2)$-free graphs.

\section{Weighted Subset Odd Cycle Transversal on ${\mathbf{(3P_1+P_2)}}$-free Graphs}\label{s-example}

We will prove that {\sc Weighted Subset Odd Cycle Transversal} is polynomial-time solvable for $(3P_1+P_2)$-free graphs using the framework from the previous section. 
We let $G=(V,E)$ be a $(3P_1+P_2)$-free graph with a vertex weighting~$w$, and let $T\subseteq V$.
For Step~1, we need the polynomial-time algorithm of~\cite{CHJMP18} for {\sc Odd Cycle Transversal} on $sP_2$-free graphs ($s\geq 1$), and thus on $(3P_1+P_2)$-free graphs (take $s=4$). 
The algorithm in \cite{CHJMP18} was  for the unweighted case, but it can be easily adapted for the weighted case as shown by Lemma~\ref{l-sp2} 
(see Appendix~\ref{a-sp2} for the proof).

\begin{lemma}\label{l-sp2}
For every integer $s\geq 1$, {\sc Weighted Odd Cycle Transversal} is polynomial-time solvable for $sP_2$-free graphs.
\end{lemma}

\noindent
As Step~2 is trivial, we need to focus on Step~3.
We will reduce to a classical problem, well known to be polynomial-time solvable by standard network flow techniques.

\optproblemdef{{\sc Weighted Vertex Cut}}{a graph $G=(V,E)$, two distinct 
non-adjacent terminals $t_1$ and $t_2$, and a non-negative vertex weighting~$w$.}{determine a set $S\subseteq V\setminus \{t_1,t_2\}$ of minimum weight such that  $t_1$ and $t_2$ are in different connected components of $G-S$.}
  
\noindent
For a mixed solution $S_T$, we let $O=O(B_T)$ and $R=R(B_T)$; recall that,
by the definition, 
$O\neq \emptyset$ and $R\cap T\neq \emptyset$ (see also Figure~\ref{f-OandR}). 
For our reduction to {\sc Weighted Vertex Cut}, we need some structural results first.

\medskip
\noindent
{\bf Structural Lemmas.}
As $O$ is nonempty, $G[O]$ has at least one connected component. 
We first bound the number of components of $G[O]$.

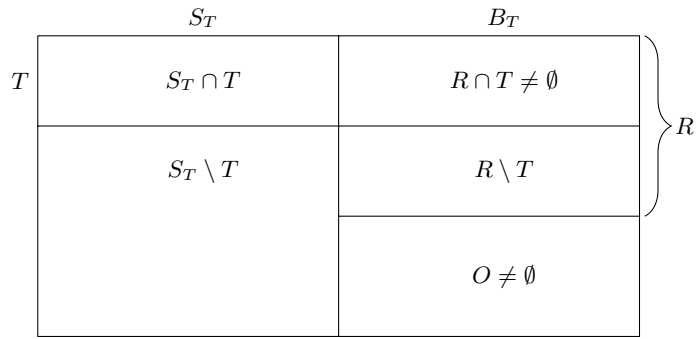
\begin{figure}
\begin{center}
\begin{tikzpicture}[scale=0.4]
\draw (-10,5)--(10,5)--(10,-5)--(-10,-5)--(-10,5) (0,5)--(0,-5) (-10,2)--(10,2) (0,-1)--(10,-1);
\node[above] at (-4.5,5) {$S_T$};
\node[above] at (5.5,5) {$B_T$};
\node[left] at (-10,3.5) {$T$};
\node at (-4.5,3.5) {$S_T\cap T$};
\node at (5.5,3.5) {$R \cap T \neq \emptyset$};
\node at (-4.5,0.5) {$S_T \setminus T$};
\node at (5.5,0.5) {$R \setminus T$};
\node at (5.5,-3) {$O \neq \emptyset$};
\draw [decorate,decoration={brace,amplitude=10pt},xshift=5pt,yshift=0pt] (10,5) -- (10,-1) node [black,midway,xshift=15pt] {$R$};

\end{tikzpicture}
\end{center}
\caption{The decomposition of $V$ when $S_T$ is a mixed solution. The sets $O=O(B_T)$ and $R=R(B_T)$ are the odd and even vertices of $B_T$, respectively.}
\label{f-OandR}
\end{figure}

\begin{lemma}\label{l-claim1}
Let $G=(V,E)$ be a $(3P_1+P_2)$-free graph, 
and let $T\subseteq V$.
For every mixed solution $S_T$, the graph $G[O]$ has at most two connected components.
\end{lemma}

\begin{proof} 
For contradiction, assume that $G[O]$ has at least three connected components $D_1$, $D_2$, $D_3$. As each $D_i$ contains an odd cycle, each $D_i$ has an edge. Hence, each $D_i$ must be a complete graph, otherwise one $D_i$, say $D_1$ has two non-adjacent vertices, which would induce together with a vertex of $D_2$ and an edge of $D_3$, a $3P_1+P_2$.

Recall that, as $S_T$ is mixed, $R$ is nonempty. Let $u\in R$. Then $u$ does not belong to any~$D_i$. Moreover, $u$ can be adjacent to at most one vertex of each $D_i$, otherwise $u$ and two of its neighbours in $D_i$ would form a triangle (as $D_i$ is complete) and $u$ would not be even. As each $D_i$ is a complete graph on at least three vertices, we can pick two non-neighbours of~$u$ in $D_1$, which form an edge, a non-neighbour of $u$ in $D_2$ and a non-neighbour of $u$ in $D_3$. These four vertices, together with $u$, induce a $3P_1+P_2$, a contradiction.  \qed
\end{proof}

\noindent
We now prove two lemmas, namely Lemmas~\ref{l-claim4} and~\ref{l-claim2}, that together show that we always have that $|R|\leq 8$.
If $G[O]$ is disconnected, then Lemma~\ref{l-claim4} proves the stronger result that $|R|\leq 2$, and if $G[0]$ is connected, we can use Lemma~\ref{l-claim2} and the fact that $G[R]$ is bipartite.

\begin{lemma}\label{l-claim4}
Let $G=(V,E)$ be a $(3P_1+P_2)$-free graph, 
and let $T\subseteq V$.
For every mixed solution $S_T$, if $G[O]$ is disconnected, then $R$ is a clique with $|R|\leq 2$.
\end{lemma}

\begin{proof}
For contradiction, suppose $R$ contains two non-adjacent vertices $u_1$ and $u_2$. Let $D$ and $D'$ be the two connected components of $G[O]$. Then $D$ has an odd cycle $C$ on vertices $v_1,\ldots,v_r$ for some $r\geq 3$ and $D'$ has an odd cycle $C'$ on vertices $w_1,\ldots,w_s$ for some $s\geq 3$.

Now, $u_1$ and $u_2$ are adjacent to at most one vertex of $C$, as otherwise they lie on an odd cycle in $G[B_T]$, which would contradict the fact that they are even vertices. Hence, as $r\geq 3$, we may assume that $v_1$ is not adjacent to $u_1$ nor to $u_2$ (see Figure~\ref{f-claim4}). Hence, at least one of $u_1$ and $u_2$ has a neighbour in $\{w_1,w_2\}$, otherwise $\{u_1,u_2,v_1,w_1,w_2\}$ would induce a $3P_1+P_2$. Say $u_1$ is adjacent to $w_1$. Similarly, one of $u_1$, $u_2$ has a neighbour in $\{w_2,w_3\}$. As~$u_1$ already has a neighbour in $C'$, we find that $u_1$ cannot be adjacent to $w_2$ or $w_3$, otherwise~$u_1$ would be in an odd cycle of $G[B_T]$, contradicting $u_1 \in R$. Hence, $u_2$ is adjacent to either~$w_2$ or $w_3$. So $u_1$ and $u_2$ each have a neighbour on $C'$ and these neighbours are not the same.

By the same reasoning, but with the roles of $C$ and $C'$ reversed, we find that $u_1$ and~$u_2$ also have (different) neighbours on $C$. However, we now find that there exists an odd cycle using $u_1$, $u_2$ and appropriate paths $P_C$ and $P_{C'}$ between their neighbours on $C$ and $C'$, respectively. We conclude that $R$ is a clique, and thus, as $G[R]$ is bipartite, $|R|\leq 2$. \qed
\end{proof}

\begin{figure}
\begin{center}
\begin{tikzpicture}[scale=0.6]
\node at (6,-1) {$C'$};
\node at (0,-1) {$C$};

\node at (13,5) {$R$};
\node at (13,-0.5) {$O$};

\draw[dashed] (2,5)--(6,5)  (6,5)--(2,0) (2,5)--(2,0);
\filldraw [black] (8,0) circle [radius=5pt] (3,-2) circle [radius=5pt] (2,-4) circle [radius=5pt];
\filldraw [black] (9,-2) circle [radius=5pt] (8,-4) circle [radius=5pt];
\draw (8,0)--(9,-2)--(8,-4) (2,0)--(3,-2)--(2,-4);

\node[left] at (2,5) {$u_1$};
\node[right] at (6,5) {$u_2$};
\node[right] at (2,0) {$v_1$};
\node[right] at (3,-2) {$v_2$};
\node[right] at (2,-4) {$v_3$};
\node[right] at (8,0) {$w_1$};
\node[right] at (9,-2) {$w_2$};
\node[right] at (8,-4) {$w_3$};

  \draw[decorate,decoration=snake] (2,0) to [out=180, in = 225](2,-4);
  \draw[decorate,decoration=snake] (8,0) to [out=180, in = 225](8,-4);

\draw [black, fill=white] (6,5) circle [radius=5pt] (2,5) circle [radius=5pt] (2,0) circle [radius=5pt];

\end{tikzpicture}
\end{center}
\caption{The situation in Lemma~\ref{l-claim4} where dotted lines indicate non-edges. Note that not all edges incident with $u_1$ and $u_2$ are drawn.}
\label{f-claim4}
\end{figure}
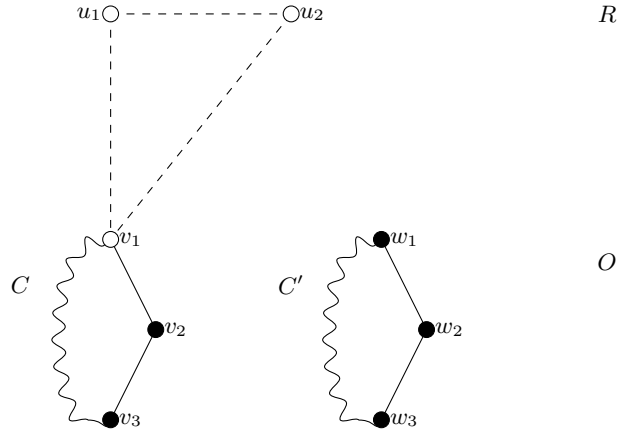

\begin{lemma}\label{l-claim2}
Let $G=(V,E)$ be a $(3P_1+P_2)$-free graph
and let $T\subseteq V$.
For every mixed solution $S_T$, every independent set in $G[R]$ has size at most~$4$.
\end{lemma}

\begin{proof}
Suppose that $R$ contains an independent set $I=\{u_1,\ldots,u_5\}$ of five vertices. As $S_T$ is mixed, $O$ is nonempty. Hence, $G[B_T]$ has an odd cycle $C$.  Let $v_1$, $v_2$, $v_3$ be consecutive vertices of $C$ in that order. As $G$ is $(3P_1+P_2)$-free, $v_1v_2\in E$ and $\{u_1,u_2,u_3\}$ is independent, one of $v_1$, $v_2$ is adjacent to one of $u_1,u_2,u_3$, say $v_1$ is adjacent to $u_1$.  Then $v_1$ must be adjacent to at least two vertices of $\{u_2,u_3,u_4,u_5\}$, otherwise three non-neighbours of $v_1$ in $\{u_2,u_3,u_4,u_5\}$, together with the edge $u_1v_1$, would induce a $3P_1+P_2$. Hence, we may assume without loss of generality that $v_1$ is adjacent to $u_2$ and $u_3$.

Let $i\in \{1,2,3\}$. As $u_i$ is adjacent to $v_1$ and $C$ is odd, $u_i$ cannot be adjacent to $v_2$ or $v_3$, otherwise $u_i$ would belong to an odd cycle in $G[B_T]$, so $u_i$ would not be even, contradicting that $u_i \in R$. Hence, $\{u_1,u_2,u_3,v_2,v_3\}$ induces a $3P_1+P_2$, a contradiction. \qed
\end{proof}

\noindent
We will now look into the ways $O$ and $R$ are connected to each other.
We say that a vertex in $O$ is a {\it connector} if it has a neighbour in $R$. 
Here is our first structural lemma on connectors.

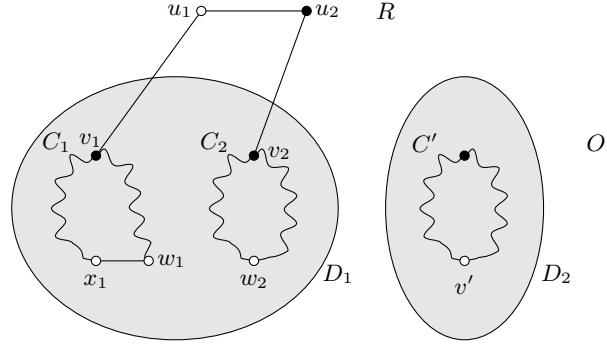
\begin{figure}
\begin{center}
\begin{tikzpicture}[scale=0.35]
\node at (9,4) {$R$};
\node at (17,-1) {$O$};
\node at (7.2,-6) {$D_1$};
\node at (15.5,-6) {$D_2$};
\draw[color=black, fill=gray!20] (1,-3.5) ellipse (6.2cm and 5cm);
\draw[color=black, fill=gray!20] (12,-3.5) ellipse (3cm and 5cm);
\filldraw [black]  (-2,-1.5) circle [radius=5pt] (6,4)  circle [radius=5pt];
\filldraw [black] (4,-1.5) circle [radius=5pt]  (12,-1.5) circle [radius=5pt];
\draw (2,4)--(6,4);
\draw (2,4)--(-2,-1.5) (4,-1.5)--(6,4);
\node[left] at (2,4) {$u_1$};
\node[right] at (6,4) {$u_2$};
\node at (2.5,-1) {$C_2$};
\node at (-3.5,-1) {$C_1$};
\node at (10.5,-1) {$C'$};
\draw  (0,-5.5)--(-2,-5.5);
\node[above, xshift=-2pt] at (-2,-1.5) {$v_1$};
\node[right, xshift=2pt, yshift=1pt] at (4,-1.5) {$v_2$};
\node[right] at (0,-5.5) {$w_1$};
\node[below, yshift=-2pt] at (-2,-5.5) {$x_1$};
\node[below, yshift=-2pt] at (4,-5.5) {$w_2$};
\node[below, yshift=-2pt] at (12,-5.5) {$v'$};
  \draw[decorate,decoration=snake] (-2,-1.5) to [out=180, in =180](-2,-5.5);
  \draw[decorate,decoration=snake] (-2,-1.5) to [out=0, in = 135](0,-5.5);
  \draw[decorate,decoration=snake] (4,-1.5) to [out=180, in = 180](4,-5.5);
  \draw[decorate,decoration=snake] (4,-1.5) to [out=0, in = 0](4,-5.5);
\draw[decorate,decoration=snake] (12,-1.5) to [out=180, in = 180](12,-5.5);
  \draw[decorate,decoration=snake] (12,-1.5) to [out=0, in = 0](12,-5.5);
\draw [black, fill=white]  (2,4) circle [radius=5pt] (0,-5.5) circle [radius=5pt] (-2,-5.5) circle [radius=5pt] (4,-5.5) circle [radius=5pt] (12,-5.5) circle [radius=5pt];

\end{tikzpicture}
\end{center}
\caption{An illustration for the proof of Lemma~\ref{l-claim5}: the white vertices induce a $3P_1+P_2$.} 
\label{f-claim5}
\end{figure}

\begin{lemma}\label{l-claim5}
Let $G=(V,E)$ be a $(3P_1+P_2)$-free graph,  
and let $T\subseteq V$.
For every mixed solution $S_T$, if $G[O]$ has two connected components $D_1$ and $D_2$, then $D_1$ and $D_2$ each have at most one connector.
\end{lemma}

\begin{proof}
By Lemma~\ref{l-claim4}, $R$ is a clique of size at most~$2$.
For contradiction, suppose that, say, $D_1$ has two distinct connectors $v_1$ and $v_2$. Then $v_1$ and $v_2$ each have at most one neighbour in~$R$, else the vertices of $R$ would be in an odd cycle in $G[B_T]$, as $R$ is a clique. Let $u_1$ be the neighbour of $v_1$ in $R$, and let $u_2$ be the neighbour of $v_2$ in $R$; note that $u_1=u_2$ is possible.
 
An edge on a path $P$ from $v_1$ to $v_2$ in $D_1$ does not belong to an odd cycle in $G[D_1]$; else there would be a path $P'$ from $v_1$ to $v_2$ in $G[O]$ with a different parity than $P$ and one of the cycles $u_1v_1Pv_2u_2u_1$ or $u_1v_1P'v_2u_2u_1$ is odd, implying that $u_1$ and $u_2$ would not be even.
 
By definition, $v_1$ and $v_2$  belong to at least one odd cycle, which we denote by $C_1$ and $C_2$, respectively.
Then $V(C_1)\cap V(C_2)=\emptyset$ and there is no edge between a vertex of $C_1$ and a vertex of $C_2$ except from possibly the edge $v_1v_2$; else there would be a path from $v_1$ to $v_2$ in $G[O]$ with an edge that belongs to an odd cycle ($C_1$ or $C_2$), a contradiction with what we found above. Note also that $u_1$ has no neighbours in $V(C_1)$ other than $v_1$; otherwise $G[B_T]$ would have an odd cycle containing $u_1$. Moreover, $u_1$ has no neighbours in $V(C_2)$ either, except $v_2$ if $u_1=u_2$; otherwise $G[B_T]$ would contain an odd cycle containing $u_1$ and $u_2$.
 
We now let $w_1$ and $x_1$ be two adjacent vertices of $C_1$ that are not adjacent to $u_1$. Let $w_2$ be a vertex of $C_2$ not adjacent to $u_1$.
Then, we found that $\{u_1,w_2,w_1,x_1\}$ induces a $2P_1+P_2$ (see Figure~\ref{f-claim5}).
 
We continue by considering $D_2$, the other connected component of $G[O]$. By definition, $D_2$ has an odd cycle~$C'$. As $|R|
\leq 2$ and each vertex of $R$ can have at most one neighbour on an odd cycle in $G[B_T]$, we find that $C'$
contains a vertex~$v'$ not adjacent to any vertex of $R$, so $v'$ is not adjacent to $u_1$. As $v'$ and the vertices of $\{w_2,w_1,x_1\}$ belong to different connected components of $G[O]$, we find that $v'$ is not adjacent to any vertex of $\{w_2,w_1,x_1\}$ either.
However, now $\{u_1,v',w_2,w_1,x_1\}$ induces a $3P_1+P_2$ (see also Figure~\ref{f-claim5}), a contradiction. \qed
\end{proof}

\noindent
We need one more structural lemma (Lemma~\ref{l-claim3}) about connectors, in the case where $G[O]$ is connected. 
In order to be able to make use of this lemma we need to exclude a special kind of mixed solution. 
Let $R$ consist of two adjacent vertices $u_1$ and $u_2$. Let $O$ (with $O\cap T=\emptyset$) be the disjoint union of two complete graphs $K$ and $L$, each on an odd number of vertices that is at least $3$, plus a single additional edge, 
such that:

\begin{enumerate}
\item $u_1$ is adjacent to exactly one vertex $v_1$ in $K$ and to no vertex of $L$;
\item $u_2$ is adjacent to exactly one vertex $v_2$ in $L$ and to no vertex of $K$; and
\item $v_1$ and $v_2$ are adjacent.
\end{enumerate}

\noindent
Note that $G[B_T]$ is indeed $T$-bipartite. We call the corresponding mixed solution $S_T$ a {\it 2-clique solution} (see Figure~\ref{f-2clique}).

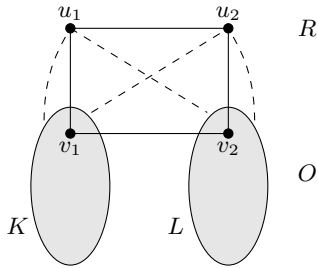
\begin{figure}
\begin{center}
\begin{tikzpicture}[scale=0.35]
\node at (11,4) {$R$};
\node at (11,-1.5) {$O$};
\draw[color=black, fill=gray!20] (2,-2) ellipse (1.5cm and 3cm);
\draw[color=black, fill=gray!20] (8,-2) ellipse (1.5cm and 3cm);
\filldraw [black] (2,4) circle [radius=5pt] (8,4) circle [radius=5pt] (2,0) circle [radius=5pt] (8,0)  circle [radius=5pt];
\draw (2,4)--(8,4)--(8,0)--(2,0)--(2,4);
\draw[dashed] (2,4)--(7.2,0.8) (2.8,0.8)--(8,4);
\draw[dashed] (1.0,0.5) to [bend left=15] (2,4);
\draw[dashed] (9.0,0.5) to [bend right=15] (8,4);
\node[above] at (2,4) {$u_1$};
\node[above] at (8,4) {$u_2$};
\node[below] at (2,0) {$v_1$};
\node[below] at (8,0) {$v_2$};
\node at (0,-3.5) {$K$};
\node at (6,-3.5) {$L$};
\end{tikzpicture}
\end{center}
\caption{The structure of $B_T$ corresponding to a 2-clique solution $S_T$. The subgraphs $K$ and $L$ are each cliques on an odd number of vertices that is at least~$3$.}
\label{f-2clique}
\end{figure}

\begin{lemma}\label{l-claim3}
Let $G=(V,E)$ be a $(3P_1+P_2)$-free graph 
and let $T\subseteq V$.
For every mixed solution $S_T$ that is not a $2$-clique solution, if $G[O]$ is connected, then $O$ has no two connectors with a neighbour in the same connected component of~$G[R]$.
 \end{lemma}

\begin{proof}
For some $p\geq 1$,  let $F_1,\ldots,F_p$ be the set of components of $G[R]$.
 For contradiction, assume $O$ has two distinct connectors $v_1$ and $v_2$, each with a neighbour in the same $F_i$, say, $F_1$. Let $u_1,u_2\in V(F_1)$ be these two neighbours, where $u_1=u_2$ is possible. Let $Q$ be a path from 
$u_1$ to $u_2$
 in $F_1$ (see Figure~\ref{f-PQ}). We make an important claim:
{\it All paths from $v_1$ to $v_2$ in $G[O]$ have the same parity.}
The reason is that if there exist paths $P$ and $P'$ from $v_1$ to $v_2$ in $G[O]$ that have different parity, then either the cycle $u_1v_1Pv_2u_2Qu_1$ or the cycle $u_1v_1P'v_2u_2Qu_1$ is odd. This would mean that $u_1$ and $u_2$ are not even. 

By definition, $v_1$ and $v_2$ each belong to at least one odd cycle, which we denote by $C_1$ and $C_2$, respectively. 
We choose $C_1$ and $C_2$ such that they have minimum length.
We note that $V(C_1)\cap V(C_2)=\emptyset$ and that there is no edge between a vertex of $C_1$ and a vertex of $C_2$ except possibly the edge $v_1v_2$; otherwise there would be paths from $v_1$ to $v_2$ in $G[O]$ that have different parity, a contradiction with the claim  above.

We also note that $v_1$ is the only neighbour of $u_1$ on $C_1$; otherwise $u_1$ would belong to an odd cycle of $G[B_T]$. 
Similarly, $v_2$ is the only neighbour of $u_2$ on $C_2$. Moreover, $u_1$ has no neighbour on $C_2$ except $v_2$ if $u_1=u_2$, and $u_2$ has no neighbour on $C_1$ except $v_1$ if $u_1=u_2$. 
This can be seen as follows. 
For a contradiction, first suppose that, say, $u_1$ has a neighbour~$w$ on $C_2$ and $w\neq v_2$. As $C_2$ is an odd cycle, there exist two vertex-disjoint paths $P$ and $P'$ on $C_2$ from $w$ to $v_2$ of different parity. Using the edges $u_1w$ and $u_2v_2$ and the path $Q$ from $u_1$ to $u_2$, this means that $u_1$ and $u_2$ are on odd cycle of $G[B_T]$. However, this is not possible as $u_1$ and $u_2$ are even. Hence, $u_1$ has no neighbour on $V(C_2)\setminus \{v_2\}$. By the same reasoning, $u_2$ has no neighbour on $V(C_1)\setminus \{v_1\}$.
Now suppose that $u_1$ is adjacent to $v_2$ and that $u_1\neq u_2$. Then $u_1$ is not adjacent to $u_2$, otherwise the vertices $u_1$, $u_2$ and $v_2$ would form a triangle, and consequently, $u_1$ and $u_2$ would not be even. Recall that $V(C_1)\cap V(C_2)=\emptyset$ and that there is no edge between a vertex of $C_1$ and a vertex of $C_2$.
Hence, we can now take $u_1$, $u_2$, a vertex of $V(C_1)\setminus \{v_1\}$, and two adjacent vertices of $V(C_2)\setminus \{v_2\}$ (which exist as $C_2$ is a cycle) to find an induced $3P_1+P_2$, a contradiction. 

We now claim that $C_1$ and $C_2$ each have exactly three vertices.  For contraction, assume that at least one of them, $C_1$ has length at least~$5$ and that in~$C_1$, we have that $x$ and $y$ are the two neighbours of $v_1$. As $C_1$ has minimum length, $x$ and $y$ are not adjacent.
Let $t_1$ and $t_2$ be adjacent vertices of $C_2$ distinct from $v_2$.   Then $\{u_1,x,y,t_1,t_2\}$ induces a $3P_1+P_2$ in $G$, a contradiction. Hence, $C_1$ and $C_2$ are triangles, say with vertices $v_1$, $w_1$, $x_1$ and $v_2$, $w_2$, $x_2$, respectively.

Now suppose $G[O]$ has a path from $v_1$ to $v_2$ on at least three vertices.  Let $s$ be the vertex adjacent to $v_1$ on this path. Then $s\notin \{w_1,x_1,w_2,x_2\}$  and $s$ is not adjacent to any vertex of $\{w_1,x_1,w_2,x_2\}$ either; otherwise 
$G[O]$ contains two paths from $v_1$ to $v_2$ that are of different parity.
As $u_1$ and $s$ are not adjacent (else $u_1$ belongs to a triangle), we find that $\{s,u_1,w_2,w_1,x_1\}$ induces a $3P_1+P_2$, a contradiction (see also Figure~\ref{f-PQ}). We conclude that as $G[O]$ is connected, $v_1$ and $v_2$ must be adjacent.

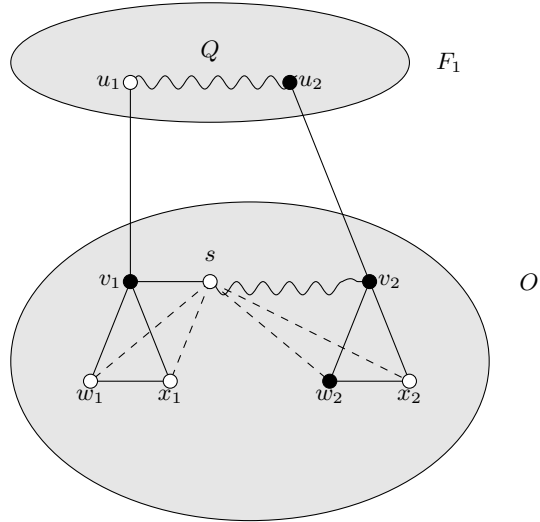
\begin{figure}
\begin{center}
\begin{tikzpicture}[scale=0.53]
\node at (10,4) {$F_1$};
\node at (12,-1.5) {$O$};
\draw[color=black, fill=gray!20] (4,4) ellipse (5cm and 1.5cm);
\draw[color=black, fill=gray!20] (5,-3.5) ellipse (6cm and 4cm);
\filldraw [black] (6,3.5) circle [radius=5pt] (7,-4) circle [radius=5pt] (2,-1.5) circle [radius=5pt] (8,-1.5)  circle [radius=5pt];
\draw (2,3.5)--(2,-1.5) (8,-1.5)--(6,3.5);
\draw (2,-1.5)--(1,-4)--(3,-4)--(2,-1.5);
\draw (8,-1.5)--(7,-4)--(9,-4)--(8,-1.5);
\draw[dashed] (4,-1.5)--(7,-4);
\draw[dashed] (4,-1.5)--(9,-4);
\draw[dashed] (4,-1.5)--(1,-4);
\draw[dashed] (4,-1.5)--(3,-4);
\node[left] at (2,3.5) {$u_1$};
\node[right] at (6,3.5) {$u_2$};
\node[left] at (2,-1.5) {$v_1$};
\node[right] at (8,-1.5) {$v_2$};
\node[above] at (4,3.8) {$Q$};
\node[below] at (1,-4) {$w_1$};
\node[below] at (3,-4) {$x_1$};
\node[below] at (7,-4) {$w_2$};
\node[below] at (9,-4) {$x_2$};
\node[above] at (4,-1.2) {$s$};
\draw  decorate [decoration={snake}] {(2,3.5)--(6,3.5) (4,-1.5)--(8,-1.5)};
\draw (2,-1.5)--(4,-1.5);
\draw [black, fill=white] (2,3.5) circle [radius=5pt]  (1,-4) circle [radius=5pt] (3,-4) circle [radius=5pt]  (9,-4)  circle [radius=5pt];
\draw [black, fill=white] (4,-1.5) circle [radius=5pt];
\end{tikzpicture}
\end{center}
\caption{The white vertices induce a $3P_1+P_2$.}
\label{f-PQ}
\end{figure}

So far, we found that $O$ contains two vertex-disjoint triangles on vertex sets $\{v_1,w_1,x_1\}$ and $\{v_2,w_2,x_2\}$, respectively, with $v_1v_2$ as the only edge between them. As $v_1$ is adjacent to $v_2$, we find that $u_1\neq u_2$; otherwise $\{u_1,v_1,v_2\}$ would induce a triangle, which is not possible as $u_1\in R$. Recall that $u_1$ is not adjacent to any vertex of $V(C_1)\cup V(C_2)$ except $v_1$, and similarly, $u_2$ is not adjacent to any vertex of $V(C_1)\cup V(C_2)$ except $v_2$. Then $u_1$ must be adjacent to $u_2$, as otherwise $\{u_1,u_2,w_1,w_2,x_2\}$ would induce a $3P_1+P_2$.

Let $z\in O\setminus (V(C_1)\cup V(C_2))$. Suppose $u_1$ is adjacent to $z$. First assume $z$ is adjacent to $w_1$ or $x_1$, say $w_1$. Then $u_1zw_1x_1v_1u_1$ is an odd cycle. Hence, this is not possible. Now assume $z$ is adjacent to $w_2$ or $x_2$, say $w_2$. Then $u_1zw_2v_2u_2u_1$ is an odd cycle. This is not possible either. 
Hence, $z$ is not adjacent to any vertex of $\{w_1,x_1,w_2,x_2\}$. Moreover, $z$ is not adjacent to $u_2$, as otherwise $\{u_1, u_2, z\}$ induces a triangle in $G[B_T]$. However, $\{u_2,w_2,z,w_1,x_1\}$ now induces a $3P_1+P_2$. Hence, $u_1$ is not adjacent to $z$. In other words, $v_1$ is the only neighbour of $u_1$ on $O$.  By the same arguments, $v_2$ is the only neighbour of $u_2$ on $O$.

Let $K$ be a maximal clique of $O$ that contains $C_1$ and let $L$ be a maximal clique of $O$ that contains~$C_2$. Note that 
$K$ and $L$ are vertex-disjoint, as for example, $w_1\in K$ and $w_2\in L$ are not adjacent.
We claim that $O=K\cup L$. For contradiction, assume that $r$ is a vertex of $O$ that does not belong to $K$ or $L$. 
As $u_1$ and $u_2$ are adjacent vertices that have no neighbours in $O\setminus \{v_1,v_2\}$, the $(3P_1+P_2)$-freeness of $G$ implies that $G[O\setminus \{v_1,v_2\}]$ is $3P_1$-free. 
As $K\setminus \{v_1\}$ and $L\setminus \{v_2\}$ induce the disjoint union of two complete graphs on at least two vertices, this means that $r$ is adjacent to every vertex of $K\setminus \{v_1\}$ or to every vertex of $L\setminus \{v_2\}$, say $r$ is adjacent to every vertex of $K\setminus \{v_1\}$. Then $r$ has no neighbour $r'$ in $L\setminus \{v_2\}$, as otherwise the cycle $v_1u_1u_2v_2r'rw_1v_1$ is an odd cycle in $G[B_T]$ that contains $u_1$ (and $u_2$).
Moreover, as $K$ is maximal and $r$ is adjacent to every vertex of $K\setminus \{v_1\}$, we find that $r$ and $v_1$ are not adjacent.
Recall also that $u_2$ has $v_2$ as its only neighbour in $O$, hence $u_2$ is not adjacent to~$r$. This means that
$\{r,v_1,u_2,w_2,x_2\}$ induces a $3P_1+P_2$, which is not possible. We conclude that $O=K\cup L$; consequently, both $K$ and $L$ have odd size.

We now consider the graph $F_1$ in more detail. Suppose $F_1$ contains another vertex $u_3\notin \{u_1,u_2\}$. As $F_1$ is connected and bipartite (as $V(F_1)\subseteq R$), we may assume without loss of generality that $u_3$ is adjacent to $u_1$ but not to $u_2$. If $u_3$ has a neighbour $K$, then $G[B_T]$ contains an odd cycle that uses $u_1$, $u_3$ and one vertex of $K$ 
 (if the neighbour of $u_3$ in $K$ is $v_1$) or three vertices of $K$ (if the neighbour of $u_3$ in $K$ is not $v_1$).
 Hence, $u_3$ has no neighbour in $K$. This means that $\{u_2,u_3,w_2,w_1,x_1\}$ induces a $3P_1+P_2$, so $u_3$ cannot exist. Hence, $F_1$ consists only of the two adjacent vertices $u_1$ and $u_2$.

Now suppose that $p\geq 2$, that is, $F_2$ is nonempty. Let $u'\in V(F_2)$. As $u'\in R$, we find that $u'$ is adjacent to at most one vertex of $C_1$ and to at most one vertex of $C_2$. Hence, we may without loss of generality assume that $u'$ is not adjacent to $w_1$ and $w_2$. Then $\{u',w_1,w_2,u_1,u_2\}$ induce a $3P_1+P_2$. We conclude that $R=\{u_1,u_2\}$. However, now $S_T$ is a 2-clique solution of $G$, a contradiction. \qed
\end{proof}

\noindent
{\bf An algorithmic lemma.} As part of our algorithm we need to be able to find a 2-clique solution of minimum weight in polynomial time.
This is shown in the next lemma.

\begin{lemma}\label{l-2clique} 
Let $G=(V,E)$ be a $(3P_1+P_2)$-free graph with a vertex weighting~$w$, and let $T\subseteq V$. It is possible to find in polynomial time a 2-clique solution for $(G,w,T)$ that has minimum weight.
\end{lemma}

\begin{proof}
As the cliques $K$ and $L$ in $B_T$ have size at least~$3$ for a 2-clique solution~$S_T$, there are distinct vertices $x_1,y_1$ in $K\setminus \{v_1\}$ and distinct vertices $x_2,y_2$ in $L\setminus \{v_2\}$.
The ordered $8$-tuple $(u_1,u_2,v_1,v_2,x_1,y_1,x_2,y_2)$ is a {\it skeleton} of the 2-clique solution.
We call the labelled subgraph of $B_T$ that these vertices induce a \emph{skeleton graph}. 
(see Figure~\ref{f-skeleton}).

\begin{figure}
\begin{center}
\begin{tikzpicture}[scale=0.45]
\filldraw [black] (2,4) circle [radius=5pt] (6,4) circle [radius=5pt] (2,0) circle [radius=5pt] (6,0)  circle [radius=5pt];
\draw (2,4)--(6,4)--(6,0)--(2,0)--(2,4);
\filldraw [black] (1,-3) circle [radius=5pt] (3,-3) circle [radius=5pt] (5,-3) circle [radius=5pt] (7,-3)  circle [radius=5pt];
\draw (2,0)--(1,-3)--(3,-3)--(2,0);
\draw (6,0)--(5,-3)--(7,-3)--(6,0);
\node[left] at (2,4) {$u_1$};
\node[right] at (6,4) {$u_2$};
\node[left] at (2,0) {$v_1$};
\node[right] at (6,0) {$v_2$};
\node[left] at (1,-3) {$x_1$};
\node[below] at (3,-3) {$y_1$};
\node[below] at (5,-3) {$x_2$};
\node[right] at (7,-3) {$y_2$};
\end{tikzpicture}
\end{center}
\caption{A skeleton graph.}
\label{f-skeleton}
\end{figure}
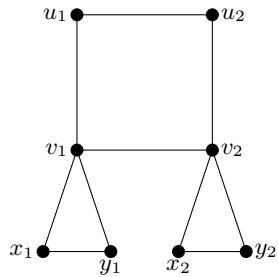

In order to find a $2$-clique solution of minimum weight in polynomial time, we consider all ${\mathcal O}(n^8)$ possible ordered $8$-tuples $(u_1,u_2,v_1,v_2,x_1,y_1,x_2,y_2)$ of vertices of $G$ and further investigate those that induce a skeleton graph.  We note that if these vertices form the skeleton of a 2-clique solution $S_T$, then $R(B_T)=\{u_1, u_2\}$ and $O(B_T)$ is a subset of $$V'=\{v_1,x_1,y_1\}\cup \{v_2,x_2,y_2\}\cup (N(v_1)\cap N(x_1)\cap N(y_1))\cup (N(v_2)\cap N(x_2)\cap N(y_2)).$$  We further refine the definition of $V'$ by deleting any vertex that cannot, by definition, belong to $O(B_T)$; that is, we remove every vertex that belongs to $T\cup (N(\{u_1,u_2\})\setminus \{v_1,v_2\})$  
or is a neighbour of both a vertex in $\{v_1,x_1,y_1\}$ and a vertex in $\{v_2,x_2,y_2\}$.
We write $G'=G[V']$. Note that $u_1$ and $u_2$ are not in $G'$
(as they are not adjacent to any vertex in $\{x_1,x_2,y_1,y_2\}$),
whereas $v_1,v_2,x_1,x_2,y_1,y_2$ all are in~$G'$.

Let $K'=\{v_1,x_1,y_1\}\cup (N(\{v_1,x_1,y_1\})\cap V')$ and $L'= \{v_2,x_2,y_2\}\cup (N(\{v_2,x_2,y_2\})\cap V')$.
We now show that 
\begin{itemize}
\item [(i)] $K'$ and $L'$ partition $V'$, and
\item [(ii)] $K'$ and $L'$ are cliques. 
\end{itemize}
By definition, every vertex of $V'$ either belongs to $K'$ or to $L'$. 
By construction, $K'\cap L'=\emptyset$ since every vertex in $K' \setminus \{v_1\}$ is a neighbour of $v_1$ and every vertex in $L' \setminus \{v_2\}$ is a neighbour of $v_2$ and no vertex in $V'$ is adjacent to both $v_1$ and $v_2$ which are themselves distinct. This shows (i).

We now prove (ii).
For a contradiction, suppose $K'$ is not a clique. Then $K'$ contains two non-adjacent vertices $t$ and~$t'$. As $K'\setminus \{v_1,x_1,y_1\}$ is complete to the clique $\{v_1,x_1,y_1\}$, we find that $t$ and $t'$ both belong to $K'\setminus \{v_1,x_1,y_1\}$.
By construction of $G'$, we find that $\{t,t'\}$ is anti-complete to $\{u_1,u_2,x_2\}$. By the definition of a skeleton, $\{u_1,u_2\}$ is anti-complete to $\{x_2\}$. Then $\{u_1,u_2,t,t',x_2\}$ induces a $3P_1+P_2$ in $G$, a contradiction. By the same arguments, $L'$ is a clique.

\medskip
\noindent
We will now continue as follows.
In $G'$ we first delete the edge $v_1v_2$.
Second, for $i\in \{1,2\}$ we replace the vertices $v_i$, $x_i$, $y_i$ by a new vertex $v_i^*$ that is adjacent precisely to every vertex that is a neighbour of at least one vertex of $\{v_i,x_i,y_i\}$ in $G'$.
This transforms the graph~$G'$ into the graph $G^*=(V^*,E^*)$. Note that in $G^*$ there is no edge between $v_1^*$ and~$v_2^*$. We give each vertex $z\in V^*\setminus \{v_1^*,v_2^*\}$ weight $w^*(z)=w(z)$, and for $i\in \{1,2\}$, we set $w^*(v_i^*)=w(v_i)+w(x_i)+w(y_i)$.
See Figure~\ref{f-2clique2}.

The algorithm will now solve {\sc Weighted Vertex Cut} on $(G^*,w^*)$ with terminals $v_1^*$ and $v_2^*$; recall that this can be done in polynomial time by standard network flow techniques. Let $S^*$ be the output. Then $G^*-S^*$ has two distinct connected components on vertex sets $K^*$ and $L^*$, respectively, with $v_1^*\in K^*$ and $v_2^*\in L^*$. We set $K=(K^*\setminus \{v_1^*\}) \cup \{v_1,x_1,y_1\}$ and $L=(L^*\setminus \{v_2^*\}) \cup \{v_2,x_2,y_2\}$ and note that $G'-S^*$ contains $G[K]$ and $G[L]$ as distinct connected components.

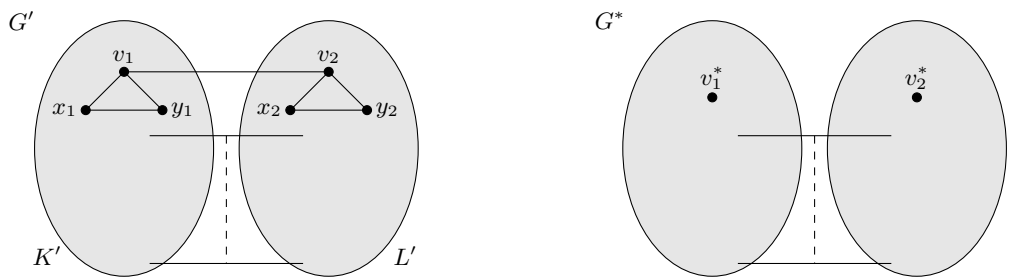
\begin{figure}
\begin{center}
\begin{tikzpicture}[scale=0.34]
\node at (-4,2) {$G'$};
\draw[color=black, fill=gray!20] (0,-3) ellipse (3.5cm and 5cm);
\draw[color=black, fill=gray!20] (8,-3) ellipse (3.5cm and 5cm);
\filldraw [black]  (0,0) circle [radius=5pt] (8,0)  circle [radius=5pt] (-1.5,-1.5) circle [radius=5pt] (1.5,-1.5)  circle [radius=5pt] (6.5,-1.5) circle [radius=5pt] (9.5,-1.5)  circle [radius=5pt];
\draw (8,0)--(0,0);
\draw (0,0)--(1.5,-1.5)--(-1.5,-1.5)--(0,0);
\draw (8,0)--(9.5,-1.5)--(6.5,-1.5)--(8,0);

\draw (1,-2.5)--(7,-2.5);
\draw[dashed] (4,-2.5)--(4,-7.5);
\draw (1,-7.5)--(7,-7.5);
\node[above] at (0,0) {$v_1$};
\node[above] at (8,0) {$v_2$};
\node[left] at (-1.5,-1.5) {$x_1$};
\node[right] at (1.5,-1.5) {$y_1$};
\node[left] at (6.5,-1.5) {$x_2$};
\node[right] at (9.5,-1.5) {$y_2$};
\node at (-3,-7.2) {$K'$};
\node at (11,-7.2) {$L'$};

\begin{scope}[xshift=23cm]
\node at (-4,2) {$G^*$};
\draw[color=black, fill=gray!20] (0,-3) ellipse (3.5cm and 5cm);
\draw[color=black, fill=gray!20] (8,-3) ellipse (3.5cm and 5cm);
\filldraw [black]  (0,-1) circle [radius=5pt] (8,-1)  circle [radius=5pt];
\draw (1,-2.5)--(7,-2.5);
\draw[dashed] (4,-2.5)--(4,-7.5);
\draw (1,-7.5)--(7,-7.5);

\node[above] at (0,-1) {$v_1^*$};
\node[above] at (8,-1) {$v_2^*$};

\end{scope}

\end{tikzpicture}
\end{center}
\caption{The graph $G'$ and $G^*$ in the proof of Lemma~\ref{l-2clique}.} 
\label{f-2clique2}
\end{figure}

As $K$ is a subset of the clique $K'$ and $L$ is a subset of the clique $L'$ and $V'=K'\cup L'$, we find that $G[K]$ and $G[L]$ are the only two connected components of $G'-S'$, and moreover that $K$ and $L$ are cliques.
As no vertex of $(K\cup L)\setminus \{v_1,v_2\}$ is adjacent to $u_1$ or~$u_2$, this means that $S=V\setminus (\{u_1,u_2\}\cup K\cup L)$ is a 2-clique solution for $G$. Moreover, as $S^*$ is an optimal solution of {\sc Weighted Vertex Cut} on instance $(G^*,w^*)$ with terminals $v_1^*$ and $v_2^*$, we find that $S$ has minimum weight over all 2-clique solutions with skeleton $(u_1,u_2,v_1,v_2,x_1,y_1,x_2,y_2)$. 

From all the ${\mathcal O}(n^8)$ 2-clique solutions computed in this way, we pick one with minimum weight; 
note that 
we found this 2-clique solution in polynomial time. \qed
\end{proof}

\noindent
{\bf The Algorithm.} We are now ready to prove the main result of the section.\\[-15pt]
\begin{theorem}\label{t-3p1p2}
{\sc Weighted Subset Odd Cycle Transversal} is polynomial-time solvable for $(3P_1+P_2)$-free graphs.
\end{theorem}

\begin{proof}
Let $G$ be a $(3P_1+P_2)$-free graph with a vertex weighting~$w$, and let $T\subseteq V(G)$.
We describe a polynomial-time algorithm for the optimization version of the problem on input $(G,T,w)$ using the approach of Section~\ref{s-general}.
So, in Step~1, we compute a neutral solution of minimum weight, i.e., a minimum weight odd cycle transversal, using
polynomial time due to Lemma~\ref{l-sp2} (take $s=4$).
We then compute, in Step~2, a $T$-full solution by setting $S_T=T$.
It remains to compute a mixed solution $S_T$ of minimum weight (Step~3) and compare its weight with the two solutions found above (Step~4).
By Lemma~\ref{l-claim1} we can distinguish between two cases: $G[O]$ is connected or $G[O]$ consists of two connected components. 
We compute a mixed solution of minimum weight for each type.

\medskip
\noindent
{\bf Case 1.} $G[O]$ is connected.\\
We first compute in polynomial time a 2-clique solution of minimum weight by using Lemma~\ref{l-2clique}.
In the remainder of Case~1, we will compute a mixed solution $S_T$ of minimum weight with connected $G[O]$ that is not a 2-clique solution.

By Lemma~\ref{l-claim2} and the fact that $G[R]$ is bipartite by definition, we find that $|R|\leq 8$. We consider all ${\mathcal O}(n^8)$ possibilities for $R$. We discard a choice for $R$ if $G[R]$ is not bipartite.
If $G[R]$ is bipartite, we compute a solution $S_T$ of minimum weight such that $B_T$ contains $R$.
Let $F_1,\ldots,F_p$ be the components of $G[R]$. By definition, $p\geq 1$. By Lemma~\ref{l-claim2} $p\leq 4$. By Lemma~\ref{l-claim3}, $O$ has at most $p\leq 4$ connectors. 

We now consider all ${\mathcal O}(n^4)$ possible choices for a set $D$ of at most four connectors. For each set $D$, we first check that $G[D\cup R]$ is $T$-bipartite and that there are no two vertices in $D$ with a neighbour in the same $F_i$; if one of these conditions is not satisfied, we discard our choice of $D$. If both conditions are satisfied we put the vertices of $D$ in  $O$, together with any vertex that is not in $T$ and that is not adjacent to any vertex of $R$. Then, as $G[D\cup R]$ is $T$-bipartite and no two vertices in $D$ are adjacent to the same component $F_i$, the graph $G[R\cup O]$ is $T$-bipartite. We remember the weight of $S_T=V\setminus (R\cup O)$.

In doing the above, we may have computed a set~$O$ that is disconnected or that contains even vertices. So we might compute some solutions more than once. However, we can compute each solution in polynomial time, and the total number of solutions we compute in Case~1 is 
${\mathcal O}(n^8)\cdot {\mathcal O}(n^4)={\mathcal O}(n^{12})$,
which is polynomial as well. Out of all the 2-clique solutions and other mixed solutions we found, we pick a solution $S_T=V_T\setminus (R\cup O)$ with minimum weight as the output for Case~1.

\medskip
\noindent
{\bf Case 2.} $G[O]$ consists of two connected components $D_1$ and $D_2$.\\
By Lemma~\ref{l-claim4}, $R$ is a clique of size at most~$2$.
We consider all possible ${\mathcal O}(n^2)$ options for $R$. Each time $R$ is a clique, we proceed as follows.
By Lemma~\ref{l-claim5}, both $D_1$ and $D_2$ have at most one connector.
We consider all ${\mathcal O}(n^2)$  ways of choosing at most one connector from each of them. If we choose two, they must be non-adjacent.
We discard the choice if the subgraph of~$G$ induced by $R$ and the chosen connector(s) is not $T$-bipartite. Otherwise we continue.
If we chose at most one connector $v$, we let $O$ consist of~$v$ and all vertices that do not belong to~$T$ and that do not have a neighbour in~$R$. Then $G[R\cup O]$ is $T$-bipartite and we store $S_T=V\setminus (R\cup O)$. Note that $O$ might not induce two connected components consisting of odd vertices, so we may duplicate some work. However, $R\cup O$ induces a $T$-bipartite graph and we found $O$ in polynomial time, and this is what is relevant (together with the fact that we only use polynomial time).

In the case where the algorithm chooses two (non-adjacent) connectors $v$ and~$v'$ we proceed as follows. We remove any vertex from $T$ and any neighbour of $R$ other than $v$ and~$v'$. 
Let $(G',w')$ be the resulting weighted graph 
(where $w'$ is the restriction of $w$ to $V(G')$).
We then 
solve {\sc Weighted Vertex Cut} in polynomial time on $G'$, $w'$ and with $v$ and $v'$ as terminals. Let $S$ be the output. We let $O=V(G')-S$. 
Note that $G[O]$ consists of two connected components (together with the fact that $G[R\cup \{v,v'\}]$ is $T$-bipartite, this implies that $G[R\cup O]$ is $T$-bipartite) but $G[O]$ might contain even vertices. However, what is relevant is that $G[R\cup O]$ is $T$-bipartite, and that we found $O$ in polynomial time. We remember
the solution $S_T=V\setminus (R\cup O)$.
In the end we remember from all the solutions we computed one with minimum weight as the output for Case~2. Note that 
the number of solutions is ${\mathcal O}(n^2)\cdot {\mathcal O}(n^2)={\mathcal O}(n^4)$ and we found each solution in polynomial time. Hence, processing Case~2 takes polynomial time.

\medskip
\noindent
{\bf Correctness and Running Time.}
The correctness of our algorithm follows from the correctness of Cases~1 and~2, which describe all possible mixed solutions due to Lemma~\ref{l-claim1}.
As processing Cases~1 and~2 takes polynomial time, we compute a mixed solution of minimum weight in polynomial time. 
Computing a non-mixed solution of minimum weight takes polynomial time as deduced already. Hence, the running time is polynomial. \qed
\end{proof}

\section{Weighted Subset Odd Cycle Transversal on $(P_1+P_3)$-free graphs}\label{p4+p1p3}

In this section, we will prove that {\sc Weighted Subset Odd Cycle Transversal} can be solved in polynomial time for $(P_1+P_3)$-free graphs.
We will follow the framework of Section~\ref{s-general} but in a less strict sense.  First we require some further definitions and preliminary results.

A subgraph $H$ of $G$ is a \emph{co-component} of $G$ if $H$ is a connected component of $\overline{G}$.
The \emph{closed neighbourhood} of $u$ is $N_G(u)\cup \{u\}$, which we denote by $N_G[u]$.
We omit subscripts when there is no ambiguity.
 We say that a set $X \subseteq V(G)$ \emph{meets} a subgraph $H$ of $G$ if $X \cap V(H) \neq \emptyset$.

In order to prove the result, on occasion we reduce to solving a weighted subset variant of the well-known {\sc Independent Set} problem and use the lemma below (Lemma~\ref{l-wsis}).
  We say that $I_T \subseteq V(G)$ is a \emph{$T$-independent set} of $G$ if each vertex of $I_T \cap T$ is an isolated vertex in $G[I_T]$.
  Note that $I_T$ is a $T$-independent set if and only if $V(G) \setminus I_T$ is a $T$-vertex cover.

\problemdef{{\sc Weighted Subset Independent Set}}{a graph $G$, a subset $T\subseteq V(G)$, a non-negative vertex weighting~$w$ and an integer $k\geq 1$.}{does $G$ have a $T$-independent set $I_T$ with $w(I_T) \geq k$?} 

\begin{lemma}\label{l-wsis}
  {\sc Weighted Subset Independent Set} is polynomial-time solvable for $3P_1$-free graphs.
\end{lemma}

\begin{proof}
  Let $G$ be a $3P_1$-free graph, and let $T \subseteq V(G)$.
  Suppose $I_T$ is a $T$-independent set of $G$.
  Observe that $|I_T \cap T| \le 2$: if $I_T$ contained three vertices of $T$, then they would form an independent set of size 3, contradicting that $G$ is $3P_1$-free.
  Moreover, if $|I_T \cap T| = 2$, then $|I_T| = 2$, since $I_T$ is $T$-independent and $G$ is $3P_1$-free.

  Suppose $|I_T \cap T| = 1$.
  We claim that in this case $I_T$ consists of a single vertex $t \in T$ and a clique $C \subseteq V(G) \setminus T$, where $t$ is anti-complete to $C$.
  Let $I_T \cap T = \{t\}$, say.
  Since $I_T$ is $T$-independent, $I_T \setminus \{t\} \subseteq V(G) \setminus N[t]$.  Since $G$ is $3P_1$-free, $V(G) \setminus N[t]$ is a clique, thus proving the claim.
  So there are three cases.\footnote{These cases correspond to the general framework of Section 2; Case~1 is about computing a $T$-full solution,  Case~2 is about computing a mixed solution and Case~3 is about computing a neutral solution.}
  
\medskip
\noindent
\textbf{Case 1:} $I_T \cap T = \emptyset$.

\smallskip
\noindent
\textbf{Case 2:} $|I_T \cap T| = 1$, in which case $I_T = \{t\} \cup C$ where $t \in T$ and $C$ is a clique of $V(G) \setminus T$ such that $t$ is anti-complete to $C$.  Moreover, since $I_T$ is $T$-independent, $C \subseteq V(G) \setminus N[t]$. 

\smallskip
\noindent   
\textbf{Case 3:} $|I_T \cap T| = 2$, in which case $|I_T| = 2$.

\medskip
\noindent
  We compute a collection of ${\mathcal O}(n^2)$ $T$-independent sets, and then output a set of maximum weight.  We compute the collection of $T$-independent sets as follows:

\medskip
\noindent
 \textbf{Case 1:} Set $I_T = V(G) \setminus T$.
 
\smallskip
\noindent
\textbf{Case 2:} For each vertex $t \in T$, let $U = V(G) \setminus N[t]$ and set $I_T = \{t\} \cup (U \setminus T)$.

\smallskip
\noindent
\item \textbf{Case 3:} For every pair of distinct vertices $t_1$ and $t_2$ in $T$, if $t_1$ and $t_2$ are non-adjacent, set $I_T = \{t_1,t_2\}$.

\medskip
\noindent
  By the foregoing, this collection will contain a maximum-weight independent set.
  So among these ${\mathcal O}(n^2)$ $T$-independent sets $I_T$, we output one of maximum weight. \qed
\end{proof}

\noindent
The {\it paw} is the graph obtained from a triangle after adding a new vertex that is adjacent to only one vertex of the triangle.
Alternatively, the paw is the complement of $P_1+P_3$ and is therefore denoted $\overline{P_1+P_3}$.
We need the following characterization of paw-free graphs due to Olariu~\cite{Ol88}. 

\begin{lemma}[\cite{Ol88}]\label{l-ol}
Every connected $(\overline{P_1+P_3})$-free graph is either triangle-free or $({P_1+P_2})$-free.
\end{lemma}

\noindent
We are now ready to prove the result.

\begin{theorem}\label{t-www}
{\sc Weighted Subset Odd Cycle Transversal} is polynomial-time solvable for $(P_1+P_3)$-free graphs.
\end{theorem}

\begin{proof}
  Let $G$ be a $(P_1+P_3)$-free graph.
  We present a polynomial-time algorithm for the optimization problem, where we seek to find $S_T \subseteq V(G)$ such that $S_T$ is a minimum-weight odd $T$-cycle transversal.
  Note that for such an $S_T$, the set $B_T = V(G) \setminus S_T$ is a maximum-weight set such that $G[B_T]$ is a $T$-bipartite graph.

  In $\overline{G}$, each connected component $D$ is $(\overline{P_1+P_3})$-free.
  By Lemma~\ref{l-ol}, $D$ is either triangle-free or $(P_1+P_2)$-free in $\overline{G}$; that is, $D$ is $3P_1$-free or $P_3$-free in $G$.
  Let $D_1, D_2, \dotsc, D_\ell$ be the co-components of $G$.

  Let $B_T \subseteq V(G)$ such that $G[B_T]$ is a $T$-bipartite graph.  For now, we do not require that $B_T$ has maximum weight.
  We start by considering some properties of such a set $B_T$.
  Observe that $G - T$ is a $T$-bipartite graph, so we may have $B_T \cap T = \emptyset$.

\smallskip
\noindent
{\it Claim~1. If $B_T \cap T \neq \emptyset$, then $B_T \subseteq V(D_i) \cup V(D_j)$ for some $i,j \in \{1,2,\dotsc,\ell\}$.}

\medskip
\noindent
We prove Claim~1 as follows.
  Suppose $u \in B_T \cap T$, say $u \in V(D_i)$ for some $i\in \{1,\ldots,r\}$.
  The claim holds if $B_T \subseteq V(D_1)$, so suppose $v \in B_T \setminus V(D_i)$.
  Then $v \in V(D_j)$ for some $j \in \{2,\dotsc,\ell\}$ with $j\neq i$.
  If $B_T$ also contains a vertex $v' \in D_{j'}$ for some $j' \in \{2,\dotsc,\ell\} \setminus \{i,j\}$, then $\{u,v,v'\}$ induces a triangle in $G$, since $D_i$, $D_j$, and $D_{j'}$ are co-components.
  As this triangle contains $u \in T$, it is an odd $T$-cycle of $G[B_T]$, a contradiction. \dia
  
\medskip
\noindent
Note that Claim~1 implies that $B_T$ meets at most two co-components of $G$ when $B_T \cap T \neq \emptyset$.  The next two claims consider the case when $B_T$ meets precisely two co-components of $G$.

\medskip
\noindent
{\it Claim~2. Suppose $B_T \cap T \neq \emptyset$ and there exist distinct $i,j \in \{1,\dotsc,\ell\}$ such that $B_T \cap V(D_i) \neq \emptyset$ and $B_T \cap V(D_j) \neq \emptyset$.
If $B_T \cap V(D_i)$ contains a vertex of $T$, then $B_T \cap V(D_j)$ is an independent set.}

\medskip
\noindent
We prove Claim~2 as follows.  Suppose it is not true.  Then, without loss of generality, $G[B_T \cap V(D_1)]$ contains an edge $u_1v_1$, and $B_T \cap V(D_2)$ contains a vertex $t \in T$.
  But then $\{u_1,v_1,t\}$ induces a triangle of $G$, since $V(D_1)$ is complete to $V(D_2)$, so $G[B_T]$ contains a contradictory odd $T$-cycle.
  \dia

\medskip
\noindent
{\it Claim~3. Suppose $B_T \cap T \neq \emptyset$ and there exist distinct $i,j \in \{1,\dotsc,\ell\}$ such that $B_T \cap V(D_i) \neq \emptyset$ and $B_T \cap V(D_j) \neq \emptyset$.
  Either
  \begin{itemize}
    \item $B_T \cap V(D_i)$ and $B_T \cap V(D_j)$ are independent sets of $G$, or
    \item $|B_T \cap V(D_i)| = 1$ and $B_T \cap V(D_i) \cap T = \emptyset$, and $B_T \cap V(D_j)$ is a $T$-independent set, up to swapping $i$ and $j$.
  \end{itemize}}

\medskip
  \noindent
  We prove Claim~3 as follows.
  Suppose $B_T$ meets $D_1$ and $D_2$, but $G[B_T \cap V(D_1)]$ contains an edge $u_1v_1$.
  Then, by Claim~2, $B_T \cap V(D_2)$ is disjoint from $T$.
  But $B_T \cap T \neq \emptyset$, so $B_T \cap V(D_1)$ contains some $t \in T$.
  Again by Claim~2, $B_T \cap V(D_2)$ is independent.
  It remains to show that $B_T \cap V(D_1)$ is a $T$-independent set and that $|B_T \cap V(D_2)| = 1$.
  Suppose $B_T \cap V(D_1)$ contains an edge $tw_1$, where $t \in T$.
  Then for any vertex $w_2 \in B_T \cap V(D_2)$, we have that $\{t,w_1,w_2\}$ induces a triangle, so $G[B_T]$ has a contradictory odd $T$-cycle.
  We deduce that each vertex of $T$ in $B_T \cap V(D_1)$ is isolated in $G[B_T \cap V(D_1)]$.
  Now suppose there exist distinct $w_2,w_2' \in B_T \cap V(D_2)$.
  Then $tw_2u_1v_1w_2't$ is an odd $T$-cycle, a contradiction.
  So $|B_T \cap V(D_i)| = 1$.
  \dia

\medskip
\noindent
  We now describe the polynomial-time algorithm.
  Our strategy is to compute, in polynomial time, a collection of ${\mathcal O}(n^2)$ sets $B_T$ such that $G[B_T]$ is $T$-bipartite, where a maximum-weight $B_T$ is guaranteed to be in this collection.
  It then suffices to output a set from this collection of maximum weight.

  First, we compute the co-components $D_1, D_2, \dotsc, D_\ell$ of $G$.
  For each of the $\ell={\mathcal O}(n)$ co-components, we can recognise whether it is $P_4$-free in linear time.
  If it is not, then it is also not $P_3$-free, so it is $3P_1$-free, by Lemma~\ref{l-ol}.
  If it is, then we can compute if it has an independent set of size at least three in linear time.
  Thus we determine whether the co-component is $3P_1$-free, or $P_3$-free.

  Now, for each co-component $D$, we solve {\sc Weighted Subset Odd Cycle Transversal} for $D$.
  Note that we can do this in polynomial time by Theorem~\ref{t-3p1p2} if $D$ is $3P_1$-free; otherwise, by Theorem~\ref{tt-p4} if $D$ is $P_3$-free.

  Now we consider each pair $\{D_1,D_2\}$ of distinct co-components.  Note there are ${\mathcal O}(n^2)$ pairs to consider.
  For each pair we will compute three sets $B_T$ such that $G[B_T]$ is $T$-bipartite.
  \begin{enumerate}
    \item We compute a maximum-weight independent set $I_1$ of $D_1$, and a maximum-weight independent set $I_2$ of $D_2$, where the weightings are inherited from the weighting~$w$ of $G$. Set $B_T = I_1 \cup I_2$.
      Note that $G[I_1 \cup I_2]$ is a complete bipartite graph, so it is certainly $T$-bipartite.
      We can compute these independent sets in polynomial time when restricted to $3P_1$- or $P_3$-free graphs (for example, see \cite{LVV14}).

    \item We select a maximum-weight vertex $v_1$ from $V(D_1) \setminus T$, and compute a maximum-weight $T$-independent set $I_2$ of $D_2$.  When $D_2$ is $3P_1$-free, we can solve this in polynomial time by Lemma~\ref{l-wsis}.
      On the other hand, when $D_2$ is $P_3$-free, we solve the complementary problem, in polynomial time, by Lemma~\ref{svc-p4}.
      Set $B_T = I_2 \cup \{v_1\}$.  Note that $G[B_T]$ is $T$-bipartite, since every vertex of $T$ has degree~1 in $G[B_T]$ (its only neighbour is $v_1$).

    \item This case is the symmetric counterpart to the previous: choose a maximum-weight vertex $v_2$ of $V(D_2) \setminus T$, compute a maximum-weight $T$-independent set $I_1$ of $D_1$, and set $B_T = I_1 \cup \{v_2\}$.
  \end{enumerate}

\noindent
  Finally, take the maximum-weight $B_T$ among the (at most) $3\binom{\ell}{2} + \ell + 1$ possibilities described, where the final possibility is that $B_T = V(G)\setminus S_T$.

\medskip
\noindent
  To prove correctness of this algorithm, suppose $B_T$ is a maximum-weight set such that $G[B_T]$ is $T$-bipartite.
  If $B_T \subseteq V(G) \setminus T$, then certainly the algorithm will either output $V(G) \setminus T$ or another solution with weight equal to $w(B_T)$.
  So we may assume that $B_T \cap T \neq \emptyset$.
  Now, by Claim~1, $B_T$ meets one or two co-components of $G$.
  If it meets exactly one co-component~$D$, then $B_T$ is a maximum-weight set such that $D[B_T]$ is $T$-bipartite, which will be found by the algorithm in the first phase.
  If it meets two co-components $D_1$ and $D_2$, then the correctness of the algorithm follows from Claim~3.
  This concludes the proof. \qed
\end{proof}

\section{The Proof of Theorem~\ref{t-realmain}}\label{s-rm}

To add to the new algorithms of the previous two sections, we need an algorithm for $P_4$-free graphs as well.
In fact we can show that {\sc Weighted Subset Odd Cycle Transversal} is polynomial-time solvable for $P_4$-free graphs
by an obvious adaptation of the
proof of the unweighted variant of {\sc Subset Odd Cycle Transversal}  from~\cite{BJPP20}. For completeness, we give the proof in Appendix~\ref{a-p4}.
However, we remark that the result  also follows from a theorem of Courcelle et al.~\cite{CMR} that shows that on graph classes of bounded clique-width, certain optimization problems have linear time algorithms.  We do not discuss the details, but it is enough to observe that $P_4$-free graphs have bounded clique-width, and that verifying a {\sc Yes} instance of the decision version of 
\textsc{Subset Odd Cycle Transversal} can be expressed in \textsf{MSO}$_1$ monadic second-order logic.

\begin{theorem}\label{tt-p4}
{\sc Weighted Subset Odd Cycle Transversal} is polynomial-time solvable for $P_4$-free graphs.
\end{theorem}

We also need one new hardness result.
The result is an analogue of Papadopoulos and Tzimas's hardness result for {\sc Weighted Subset Feedback Vertex Set} on $5P_1$-free graphs \cite[Theorem~2]{PT20}. 
The proof is essentially identical, as all the relevant $T$-cycles in the constructed {\sc Weighted Subset Feedback Vertex Set} instance are odd.
We provide the full proof in Appendix~\ref{a-5p1} for completeness.

\begin{theorem} \label{tt-5p1}
{\sc Weighted Subset Odd Cycle Transversal} is \NP-complete for $5P_1$-free graphs.
\end{theorem}

\noindent
We are now ready to prove Theorem~\ref{t-realmain}.

\medskip
\noindent
\faketheorem{Theorem~\ref{t-realmain} (restated).}
{\it 
Let $H$ be a graph with $H\notin \{2P_1+P_3,P_1+P_4,2P_1+P_4\}$. 
Then {\sc Weighted Subset Odd Cycle Transversal} on $H$-free graphs is polynomial-time solvable if 
$H\ssi 3P_1+P_2$, $P_1+P_3$, or $P_4$, and is \NP-complete otherwise.}

\begin{proof}
We first recall the result of~\cite{CHJMP18} that 
 {\sc Odd Cycle Transversal}, that is, {\sc Weighted Subset Odd Cycle Transversal} where $T=\emptyset$ and $w\equiv 1$, 
is \NP-complete on $H$-free graphs if $H$ has a~cycle~or~a~claw.
In the remaining case $H$ is a linear forest.
If $H$ contains an induced $2P_2$, then we use a result of~\cite{BJPP20}, which states that 
{\sc Subset Odd Cycle Transversal} is \NP-complete for split graphs, or equivalently, $(C_4,C_5,2P_2)$-free graphs.
If $H$ contains an induced $5P_1$, then we use Theorem~\ref{tt-5p1}.
In the other cases, we use Theorems~\ref{t-3p1p2},~\ref{t-www} and~\ref{tt-p4}. \qed
\end{proof}

\section{Weighted Subset Feedback Vertex Set for $(3P_1+P_2)$-free graphs}\label{a-easier}

Again we follow the framework of Section~\ref{s-general}.
Recall that a  subgraph of a graph $G=(V,E)$ is a  $T$-forest if it has no $T$-cycles. Recall also that a subset $S_T\subseteq V$ is a $T$-feedback vertex set if and only if $G[V\setminus S_T]$ is a $T$-forest and that we write $F_T=V\setminus S_T$ in this case. 
Recall that if $u\in F_T$ belongs to at least one cycle of $G[F_T]$, then $u$ is a cycle vertex of $F_T$ and otherwise $u$ is forest vertex of $F_T$, and that by definition every vertex in $T\cap F_T$ is a forest vertex.
We let $O(F_T)$ denote the set of cycle vertices of $F_T$ and $R(F_T)$ denote the set of forest vertices of $F_T$ and note that
$V(F_T)=O(F_T)\cup R(F_T)$.

Let $(G,T,w)$ be an instance of {\sc Weighted Subset Feedback Vertex Set}, where $G$ is $3P_1+P_2)$-free graph. 
For Step~1 in the framework of Section~\ref{s-general}, we must compute, in polynomial time, a neutral solution, that is, a minimum weight feedback vertex set.
For this we can use a straightforward generalization of the result of~\cite{CHJMP18} that states that the unweighted variant 
is polynomial-time solvable on $sP_2$-free graphs for every integer $s\geq 1$, or we can use the following recent result.

\begin{theorem}[\cite{PPR21}]\label{t4}
For every integer $s\geq 1$, {\sc Weighted Feedback Vertex Set} is polynomial-time solvable for $sP_3$-free graphs.
\end{theorem}

Due to Theorem~\ref{t4} we can do Step~1 in polynomial time: take $s=4$. 
As Step~2 is trivial, we need to focus on Step~3: computing a mixed solution of minimum weight. For this step, we will again reduce to
{\sc Weighted Vertex Cut}.

For a mixed solution $S_T$, we write $O=O(F_T)$ and $R=R(F_T)$ and note that $O$ and $R\cap T$ are both nonempty (by definition).
We need a number of structural results on mixed solutions, which we will present in a sequence of lemmas.\footnote{These proofs are very similar to analogous lemmas from Section~\ref{s-example} but they are simpler. Moreover, we do not need to translate all  lemmas, as we do not need to deal with 2-clique solutions either.}

\begin{lemma}\label{l-claim1f}
Let $G=(V,E)$ be a $(3P_1+P_2)$-free graph with a vertex weighting~$w$, and let $T\subseteq V$.
For every mixed solution $S_T$, the graph $G[O]$ has at most two connected components.
\end{lemma}

\begin{proof}
For contradiction, assume that $G[O]$ has at least three connected components $D_1$, $D_2$, $D_3$. As each $D_i$ contains a cycle, each $D_i$ has an edge. Hence, each $D_i$ must be a complete graph; otherwise one $D_i$, say $D_1$ has two non-adjacent vertices, which would induce together with a vertex of $D_2$ and an edge of $D_3$, a $3P_1+P_2$.

Recall that, as $S_T$ is mixed, $R$ is nonempty. Let $u\in R$. Then $u$ does not belong to any $D_i$.
Moreover, $u$ can be adjacent to at most one vertex of each $D_i$; otherwise $u$ and two of its neighbours of $D_i$ would form a triangle (as $D_i$ is complete) and $u$ would not be a forest vertex. As each $D_i$ is a complete graph on at least three vertices, we can pick two non-neighbours of $u$ in $D_1$, which form an edge, a non-neighbour of $u$ in $D_2$ and a non-neighbour of $u$ in $D_3$. These four vertices induce, together with $u$, a $3P_1+P_2$, a contradiction. \qed
\end{proof}

\begin{lemma}\label{l-claim4fis3}
Let $G=(V,E)$ be a $(3P_1+P_2)$-free graph with a vertex weighting~$w$, and let $T\subseteq V$.
For every mixed solution $S_T$, if $G[O]$ is disconnected, then $R$ is a clique with~$|R|\leq 2$.
\end{lemma}

\begin{proof}
For contradiction, suppose $R$ contains two non-adjacent vertices $u_1$ and $u_2$. As $G[O]$ is disconnected, $G[O]$ has exactly two connected components $D$ and $D'$ by Lemma~\ref{l-claim1f}.
Then $D$ contains 
a cycle $C$ on vertices $v_1,\ldots,v_r$ for some $r\geq 3$ and $D'$ contains a cycle $C'$ on vertices $w_1,\ldots,w_s$ for some $s\geq 3$.

Now, $u_1$ and $u_2$ are adjacent to at most one vertex of $C$, as otherwise they lie on a cycle in $G[F_T]$. Hence, as $r\geq 3$, we may assume that $v_1$ is neither adjacent to $u_1$ nor to $u_2$. Hence, at least one of $u_1$, $u_2$ has a neighbour in $\{w_1,w_2\}$, say $u_1$ is adjacent to $w_1$; otherwise $\{u_1,u_2,v_1,w_1,w_2\}$ would induce a $3P_1+P_2$. By the same argument, one of $u_1$, $u_2$ must have a neighbour in $\{w_2,w_3\}$. As $u_1$ already has a neighbour on $C'$, we find that $u_1$ cannot be adjacent to $w_2$ or $w_3$; otherwise $u_1$ would be on a cycle in $G[F_T]$. Hence, $u_2$ must be adjacent to either $w_2$ or $w_3$. So $u_1$ and $u_2$ each have a neighbour on $C'$ and these neighbours are not the same.

By the same reasoning, but with the roles of $C$ and $C'$ reversed, we find that $u_1$ and $u_2$ also have (different) neighbours on $C$. However, we now find that there exists a cycle using $u_1$, $u_2$ and appropriate paths $P_C$ and $P_{C'}$ between their neighbours on $C$ and $C'$, respectively. We conclude that $R$ is a clique. As $G[R]$ is bipartite, this means that $|R|\leq 2$. \qed
\end{proof} 

\begin{lemma}\label{l-claim2f}
Let $G=(V,E)$ be a $(3P_1+P_2)$-free graph with a vertex weighting~$w$, and let $T\subseteq V$.
For every mixed solution $S_T$, every independent set in $G[R]$ has size at most~$4$.
\end{lemma}

\begin{proof}
Suppose that $R$ contains an independent set $I=\{u_1,\ldots,u_5\}$ of five vertices. As $O$ is nonempty, $G[F_T]$ has a cycle $C$. 
Let $v_1$, $v_2$, $v_3$ be consecutive vertices of $C$ in that order. As $G$ is $(3P_1+P_2)$-free, $v_1v_2\in E$ and $\{u_1,u_2,u_3\}$ is independent, one of $v_1$, $v_2$ is adjacent to one of $u_1,u_2,u_3$; say $v_1$ is adjacent to $u_1$. 
Then $v_1$ must be adjacent to at least two vertices of $\{u_2,u_3,u_4,u_5\}$; otherwise three non-neighbours of $v_1$ in $\{u_2,u_3,u_4,u_5\}$, together with the edge $u_1v_1$, would induce a $3P_1+P_2$. Hence, we may assume without loss of generality that $v_1$ is adjacent to $u_2$ and $u_3$.

Let $i\in \{1,2,3\}$. As $u_i$ is adjacent to $v_1$ and $C$ is a cycle,
$u_i$ cannot be adjacent to $v_2$ or $v_3$; otherwise $u$ would belong to a cycle, so $u$ would not be a forest vertex.
Hence, $\{u_1,u_2,u_3,v_2,v_3\}$ induces a $3P_1+P_2$, a contradiction. \qed
\end{proof}

\noindent
Just as in Section~\ref{s-example}, we say that a vertex in $O$ is a {\it connector} if it has a neighbour in $R$. The proof of the following two lemmas are significantly simpler than the two analogous lemmas (Lemmas~\ref{l-claim5} and~\ref{l-claim3}) from Section~\ref{s-example}.

\begin{lemma}\label{l-claim5f}
Let $G=(V,E)$ be a $(3P_1+P_2)$-free graph with a vertex weighting~$w$, and let $T\subseteq V$.
For every mixed solution $S_T$, if $G[O]$ has two connected components $D_1$ and $D_2$, then $D_1$ and $D_2$ each have at most one connector.
\end{lemma}

\begin{proof}
By Lemma~\ref{l-claim4fis3}, we find that $R$ is a clique (of size at most~$2$). Then both $D_1$ and $D_2$ each have at most one connector; otherwise there would be a cycle passing through $R$, contradicting the fact that the vertices of $R$ are forest vertices. \qed
\end{proof}

\begin{lemma}\label{l-claim3f}
Let $G=(V,E)$ be a $(3P_1+P_2)$-free graph with a vertex weighting~$w$, and let $T\subseteq V$.
For every mixed solution $S_T$, if $G[O]$ is connected, then $O$ has no two connectors with a neighbour in the same connected component of $G[R]$.
 \end{lemma}
 
 \begin{proof}
 Let $A$ be a connected component of $G[R]$. For contradiction, assume that $O$ has two distinct connectors $v_1$ and $v_2$ that have a neighbour $u_1$ and $u_2$, respectively (possibly $u_1=u_2$) in $A$. Then we can take the cycle $u_1v_1Pv_2u_2Qu_1$, where $P$ is a path from $v_1$ to $v_2$ in $O$, and $Q$ is a path from $u_2$ to $u_1$ in $R$ implying that $u_1$ is not a forest vertex, a contradiction. \qed
 \end{proof}
 
\noindent
We are now ready to prove the result, which we do in exactly the same way as in the proof of Theorem~\ref{t-3p1p2}.

\begin{theorem}\label{t-ww}
{\sc Weighted Subset Feedback Vertex Set} is polynomial-time solvable for $(3P_1+P_2)$-free graphs.
\end{theorem}

\begin{proof}
Let $G=(V,E)$ be a $(3P_1+P_2)$-free graph with a vertex weighting~$w$, and let $T\subseteq V$.
We describe a polynomial-time algorithm for the optimization version of the problem on input $(G,T,w)$.
So, in Step~1, we first compute a neutral solution of minimum weight, that is, a minimum weight feedback vertex set.
We can do this in polynomial time by Lemma~\ref{l-sp2} (take $s=4$).
We then compute, in Step~2, a $T$-full solution by setting $S_T=T$.
It remains to compute a mixed solution $S_T$ of minimum weight (Step~3) and compare its weight with the weight of the two solutions found above (Step~4).

By Lemma~\ref{l-claim1f} we can distinguish between two cases: $G[O]$ is connected or $G[O]$ consists of two connected components. We compute a mixed solution of minimum weight for each type.

\medskip
\noindent
{\bf Case 1.} $G[O]$ is connected.\\
By Lemma~\ref{l-claim2f} and the fact that $G[R]$ is bipartite by definition, we find that $|R|\leq 8$. We consider all ${\mathcal O}(n^8)$ possibilities for $R$. For each choice of $R$ we do the following. We discard $R$ if $G[R]$ is not a forest.
If $G[R]$ is a forest, then we compute a solution $S_T$ of minimum weight such that $F_T$ contains $R$.

We let $A_1,\ldots,A_p$ be the set of connected components of $G[R]$. Note that $p\geq 1$ by definition and that $p\leq 4$ by Lemma~\ref{l-claim2f}. By Lemma~\ref{l-claim3f}, $O$ has at most $p\leq 4$ connectors. 

We now consider all ${\mathcal O}(n^4)$ possible choices for a set $D$ of at most four connectors. For each set $D$ we do as follows. We first check that $G[D\cup R]$ is a $T$-forest and that there are no two vertices in $D$ with a neighbour in the same $A_i$; if either of these two conditions is not satisfied, then we discard our choice of $D$. If both conditions are satisfied, then we put the vertices of $D$ in  $O$, together with any vertex that is not in $T$ and that is not adjacent to any vertex of $R$. Then, as $G[D\cup R]$ is a $T$-forest and no two vertices in $D$ are adjacent to the same component $A_i$, the graph $G[R\cup O]$ is a $T$-forest, and we remember the weight of $S_T=V\setminus (R\cup O)$.

We note that in the above, we may have computed a set~$O$ that is disconnected or that contains forest vertices. So our algorithm might compute some solutions more than once. However, we can compute each solution in Case~1 in polynomial time, and the total number of solutions we compute in Case~1 is ${\mathcal O}(n^8)\cdot {\mathcal O}(n^4)={\mathcal O}(n^{12})$, which is polynomial as well. Hence, out of all the mixed solutions we found so far, we can safely pick a solution $S_T=V_T\setminus (R\cup O)$ with minimum weight, as the output for Case~1.

\medskip
\noindent
{\bf Case 2.} $G[O]$ consists of two connected components $D_1$ and $D_2$.\\
By Lemma~\ref{l-claim4fis3}, the set $R$ is a clique of size at most~$2$.
We consider all possible ${\mathcal O}(n^2)$ options of choosing $R$. If $R$ is not a clique, we discard~$R$. If $R$ is a clique, then we proceed as below.

By Lemma~\ref{l-claim5}, both $D_1$ and $D_2$ have at most one connector.
We consider all ${\mathcal O}(n^2)$ possible ways to choose at most one connector for $D_1$ and at most one connector for $D_2$. If we choose two connectors, they must be non-adjacent.
We discard the choice if the subgraph of~$G$ induced by $R$ and the chosen connector(s) is not a $T$-forest. Otherwise we continue as follows.

In the case where our algorithm chooses at most one connector $v$, we let $O$ consist of $v$ and all vertices that do not belong to $T$ and that do not have a neighbour in $R$. Then $G[R\cup O]$ is a $T$-forest and we store $S_T=V\setminus (R\cup O)$. Note that $O$ might not induce two connected components consisting of cycle vertices, so we may duplicate some work. However, $R\cup O$ induces a $T$-forest, and we found $O$ in polynomial time, and this is what is relevant (together with the fact that we only use polynomial time).

In the case where the algorithm chooses two (non-adjacent) connectors $v$ and $v'$ we proceed as follows. We remove any vertex from $T$ and any neighbour of $R$ other than $v$ and~$v'$. Let $(G',w')$ be the resulting weighted graph. We then 
solve {\sc Weighted Vertex Cut} in polynomial time on $G'$, $w'$ and with $v$ and $v'$ as terminals. Let $S$ be the output. We let $O=V(G')-S$. Note that $G[O]$ consists of two connected components, but it might contain forest vertices. However, $G[R\cup O]$ is a $T$-forest, and we found $O$ in polynomial time, and this is what is relevant. We remember
the solution $S_T=V\setminus (R\cup O)$.

In the end we remember from all the solutions we computed one with minimum weight as the output for Case~2. Note that the number of solutions is ${\mathcal O}(n^2)\cdot {\mathcal O}(n^2)={\mathcal O}(n^4)$ and we found each solution in polynomial time. Hence, processing Case~2 takes polynomial time.

 \medskip
\noindent
The correctness of the algorithm now follows from the correctness of the two case descriptions, which describe all possible mixed solutions due to Lemma~\ref{l-claim1f}. As the algorithm processes each of the two cases in polynomial time, it computes a mixed solution of minimum weight in polynomial time. We already deduced that computing a non-mixed solution of minimum weight takes polynomial time as well. Hence, the total running time is polynomial. \qed
\end{proof}

\noindent
{\bf Remark.} The restriction of our algorithm for $(3P_1+P_2)$-free graphs in Theorem~\ref{t-ww} to the known case where $H=4P_1$~\cite{PT20} would allow for a slightly simpler proof of Lemma~\ref{l-claim1f} and a simpler proof of Lemma~\ref{l-claim4fis3}, whereas the proof of Lemma~\ref{l-claim2f} becomes trivial and can be improved to the immediate claim that every independent set in $G[R]$ has size at most~$3$. The (short) proofs of Lemmas~\ref{l-claim5f} and~\ref{l-claim3f} stay the same. Moreover, the proof of  Theorem~\ref{t-ww} stays the same, but due to the improvement of the statement of Lemma~\ref{l-claim2f} we obtain a better running time: the running time of Case~1 improves with a factor of ${\mathcal O}(n^3)$, as the number of solutions ${\mathcal O}(n^{12})$ in Case~1 now becomes ${\mathcal O}(n^6)\cdot {\mathcal O}(n^3)={\mathcal O}(n^9)$.

\section{Weighted Subset Feedback Vertex Set for $(P_1+P_3)$-free graphs}\label{a-p1p3f}

The proof is very similar to the proof of Theorem~\ref{t-www}, but we provide it in its entirety here, for completeness.
Note that we also use the definitions from Section~\ref{p4+p1p3}.

\begin{theorem} \label{t-fvsp1p3}
{\sc Weighted Subset Feedback Vertex Set} is polynomial-time solvable for $(P_1+P_3)$-free graphs.
\end{theorem}

\begin{proof}
  Let $G$ be a $(P_1+P_3)$-free graph.
  We present a polynomial-time algorithm for the optimization problem, where we seek to find $S_T \subseteq V(G)$ such that $S_T$ is a minimum-weight $T$-feedback vertex set.
  Note that for such an $S_T$, the set $F_T = V(G) \setminus S_T$ is a maximum-weight set such that $G[F_T]$ is a $T$-forest.

  In $\overline{G}$, each connected component $D$ is $(\overline{P_1+P_3})$-free.
  By Lemma~\ref{l-ol}, $D$ is either triangle-free or $(P_1+P_2)$-free in $\overline{G}$; that is, $D$ is $3P_1$-free or $P_3$-free in $G$.
  Let $D_1, D_2, \dotsc, D_\ell$ be the co-components of $G$.

  Let $F_T \subseteq V(G)$ such that $G[F_T]$ is a $T$-forest.  For now, we do not require that $F_T$ has maximum weight.
  We start by considering some properties of such a set $F_T$.
  Observe that $G - T$ is a $T$-forest, so we may have $F_T \cap T = \emptyset$ (that is, the solution $S_T$ is $T$-full).

\medskip
\noindent
{\it Claim~1. If $F_T \cap T \neq \emptyset$, then $F_T \subseteq V(D_i) \cup V(D_j)$ for some $i,j \in \{1,2,\dotsc,\ell\}$.}

\medskip
\noindent
We prove Claim~1 as follows.
  Suppose $u \in F_T \cap T$, say $u \in V(D_i)$ for some $i\in \{1,2,\ldots,\ell\}$.
  The claim holds if $F_T \subseteq V(D_i)$, so suppose $v \in F_T \setminus V(D_i)$.
  Then $v \in V(D_j)$ for some $j \in \{1,\dotsc,\ell\}$ with $j\neq i$.
  If $F_T$ also contains a vertex $v' \in D_{j'}$ for some $j' \in \{1,\dotsc,\ell\} \setminus \{i,j\}$, then $\{u,v,v'\}$ induces a triangle in $G$, since $D_i$, $D_j$, and $D_{j'}$ are co-components.
  As this triangle contains $u \in T$, it is a $T$-cycle of $G[F_T]$, a contradiction. \dia
  
\medskip
\noindent
Note that Claim~1 implies that $F_T$ meets at most two co-components of $G$ when $F_T \cap T \neq \emptyset$.  The next two claims consider the case when $F_T$ meets precisely two co-components of $G$.

\medskip
\noindent
{\it Claim~2. Suppose $F_T \cap T \neq \emptyset$ and there exist distinct $i,j \in \{1,\dotsc,\ell\}$ such that $F_T \cap V(D_i) \neq \emptyset$ and $F_T \cap V(D_j) \neq \emptyset$.
If $F_T \cap V(D_i)$ contains a vertex of $T$, then $F_T \cap V(D_j)$ is an independent set.}

\medskip
\noindent
We prove Claim~2 as follows. Suppose it is not true. Then, without loss of generality, $G[F_T \cap V(D_1)]$ contains an edge $u_1v_1$, and $F_T \cap V(D_2)$ contains a vertex $t \in T$.
  But then $\{u_1,v_1,t\}$ induces a triangle of $G$, since $V(D_1)$ is complete to $V(D_2)$, so $G[F_T]$ contains a contradictory $T$-cycle.
  \dia

\medskip
\noindent
{\it Claim~3. Suppose $F_T \cap T \neq \emptyset$ and there exist distinct $i,j \in \{1,\dotsc,\ell\}$ such that $F_T \cap V(D_i) \neq \emptyset$ and $F_T \cap V(D_j) \neq \emptyset$.
  Either
  \begin{enumerate}
    \item $F_T \cap V(D_i)$ and $F_T \cap V(D_j)$ are independent sets of $G$, and $|F_T \cap V(D_h)| = 1$ for some $h \in \{i,j\}$; or
    \item $|F_T \cap V(D_i)| = 1$ and $F_T \cap V(D_i) \cap T = \emptyset$, and $F_T \cap V(D_j)$ is a $T$-independent set, up to swapping $i$ and $j$.
  \end{enumerate}}

\smallskip
  \noindent
 We prove Claim~3 as follows. 
  Suppose $F_T$ meets $D_1$ and $D_2$, but $G[F_T \cap V(D_1)]$ contains an edge $u_1v_1$.
  Then, by Claim~2, $F_T \cap V(D_2)$ is disjoint from $T$.
  But $F_T \cap T \neq \emptyset$, so $F_T \cap V(D_1)$ contains some $t \in T$.
  Again by Claim~2, $F_T \cap V(D_2)$ is independent.
  We claim that $F_T \cap V(D_1)$ is a $T$-independent set and that $|F_T \cap V(D_2)| = 1$, so that (2) holds.
  Suppose $F_T \cap V(D_1)$ contains an edge $tw_1$, where $t \in T$.
  Then for any vertex $w_2 \in F_T \cap V(D_2)$, we have that $\{t,w_1,w_2\}$ induces a triangle, so $G[F_T]$ has a contradictory $T$-cycle.
  We deduce that each vertex of $T$ in $F_T \cap V(D_1)$ is isolated in $G[F_T \cap V(D_1)]$.
  Now suppose there exist distinct $w_2,w_2' \in F_T \cap V(D_2)$.
  Then $tw_2u_1v_1w_2't$ is an $T$-cycle, a contradiction.
  So $|F_T \cap V(D_1)| = 1$.  Thus (2) holds.

  Now suppose $F_T$ meets $D_1$ and $D_2$, but $F_T \cap V(D_h)$ is an independent set of $G$ for $h \in \{1,2\}$.
  We may also assume, without loss of generality, that $T \cap F_T \cap V(D_1) \neq \emptyset$, since $F_T \cap T \neq \emptyset$.
  Towards a contradiction, suppose $|F_T \cap V(D_1)|, |F_T \cap V(D_2)| \ge 2$.
  Then there exist distinct $v_1,t_1 \in F_T \cap V(D_1)$ such that $t \in T$, and there exist distinct $v_2,v_2' \in F_T \cap V(D_2)$.
  But now $t_1v_2v_1v_2't_1$ is a $T$-cycle of $G[F_T]$, a contradiction.
  We deduce that $|F_T \cap V(D_h)|=1$ for some $h \in \{i,j\}$, as required.
  \dia

\medskip
\noindent
  We now describe the polynomial-time algorithm.
  Our strategy is to compute, in polynomial time, a collection of ${\mathcal O}(n^2)$ sets $F_T$ such that $G[F_T]$ is a $T$-forest, where a maximum-weight $F_T$ is guaranteed to be in this collection.
  It then suffices to output a set from this collection of maximum weight.

  First, we compute the co-components $D_1, D_2, \dotsc, D_\ell$ of $G$.
  For each of the $\ell={\mathcal O}(n)$ co-components, we can recognise whether it is $P_4$-free in linear time.
  If it is not, then it is also not $P_3$-free, so it is $3P_1$-free, by Lemma~\ref{l-ol}.
  If it is, then we can compute if it has an independent set of size at least three in linear time.
  Thus we determine whether the co-component is $3P_1$-free, or $P_3$-free.

  Now, for each co-component $D$, we solve {\sc Weighted Subset Feedback Vertex Set} for $D$.
  Note that we can do this in polynomial by Theorem~\ref{t-ww} if $D$ is $3P_1$-free,
 and we can use the aforementioned result of Bergougnoux et al.~\cite{BPT19} for $P_4$-free graphs if $D$ is $P_3$-free.
 
  Now we consider each pair $\{D_1,D_2\}$ of distinct co-components.  Note there are ${\mathcal O}(n^2)$ pairs to consider.
  For each pair we will compute four sets $F_T$ such that $G[F_T]$ is a $T$-forest.
  \begin{enumerate}
    \item We select a maximum-weight vertex $v_1$ from $V(D_1)$, and compute a maximum-weight independent set $I_2$ of $D_2$, where the weightings are inherited from the weighting~$w$ of $G$. Set $F_T = \{v_1\} \cup I_2$.
      Note that $G[F_T]$ is a tree, so it is certainly a $T$-forest.
      We can compute these independent sets in polynomial time when restricted to $3P_1$- or $P_3$-free graphs (for example, see \cite{LVV14}).

    \item The symmetric counterpart to the previous case: we select a maximum-weight vertex $v_2$ from $V(D_2)$, and compute a maximum-weight independent set $I_1$ of $D_1$, and set $F_T = \{v_2\} \cup I_1$.

    \item We select a maximum-weight vertex $v_1$ from $V(D_1) \setminus T$, and compute a maximum-weight $T$-independent set $I_2$ of $D_2$.  When $D_2$ is $3P_1$-free, we can solve this in polynomial-time by Lemma~\ref{l-wsis}.
      On the other hand, when $D_2$ is $P_3$-free, we solve the complementary problem, in polynomial time, by Lemma~\ref{svc-p4}.
      Set $F_T = I_2 \cup \{v_1\}$.  Note that $G[F_T]$ is a $T$-forest, since every vertex of $T$ has degree~1 in $G[F_T]$ (its only neighbour is $v_1$).

    \item This case is the symmetric counterpart to the previous one: select a maximum-weight vertex $v_2$ of $V(D_2) \setminus T$, compute a maximum-weight $T$-independent set $I_1$ of $D_1$, and set $F_T = I_1 \cup \{v_2\}$.
  \end{enumerate}

\noindent
  Finally, take the maximum-weight $F_T$ among the (at most) $4\binom{\ell}{2} + \ell + 1$ possibilities described, where the final possibility is that $F_T = V(G)\setminus S_T$.

\medskip
\noindent
  To prove correctness of this algorithm, suppose $F_T$ is a maximum-weight set such that $G[F_T]$ is a $T$-forest.
  If $F_T \subseteq V(G) \setminus T$, then certainly the algorithm will either output $V(G) \setminus T$ or another solution with weight equal to $w(F_T)$.
  So we may assume that $F_T \cap T \neq \emptyset$.
  Now, by Claim~1, $F_T$ meets one or two co-components of $G$.
  If it meets exactly one co-component~$D$, then $F_T$ is a maximum-weight set such that $D[F_T]$ is a $T$-forest, which will be found by the algorithm in the first phase.
  If it meets two co-components $D_1$ and $D_2$, then the correctness of the algorithm follows from Claim~3.
  This concludes the proof. \qed
\end{proof}

\section{The Proof of Theorem~\ref{t-realmain2}}

In this section we prove Theorem~\ref{t-realmain2}.

\medskip
\noindent
\faketheorem{Theorem~\ref{t-realmain2} (restated).}
{\it Let $H$ be a graph with $H\notin \{2P_1+P_3,P_1+P_4,2P_1+P_4\}$.  Then {\sc Weighted Subset Feedback Vertex Set} on $H$-free graphs is polynomial-time solvable if $H\ssi 3P_1+P_2$, $P_1+P_3$, or $P_4$, and is \NP-complete otherwise.}

\begin{proof}
We first recall the results that {\sc Feedback Vertex Set}, that is, {\sc Weighted Subset Feedback Vertex Set} with  $T=\emptyset$ and $w\equiv 1$, is \NP-complete on $H$-free graphs if  $H$ has a cycle~\cite{Po74} or a claw~\cite{Mu17b}.
In the remaining case $H$ is a linear forest.
If $H$ contains an induced $2P_2$, then we use a result of~\cite{FHKPV14}, which states that 
{\sc Subset Feedback Vertex Set} is \NP-complete for split graphs, or equivalently, $(C_4,C_5,2P_2)$-free graphs.
If $H$ contains an induced $5P_1$, then we use Theorem~\ref{t-known}.
In the other cases, we use Theorems~\ref{t-ww},~\ref{t-fvsp1p3} and the fact that {\sc Weighted Subset Feedback Vertex Set} is polynomial-time solvable for $P_4$-free graphs~\cite{BPT19}, respectively. \qed
\end{proof}

\section{Conclusions}\label{s-conclusions}

By developing a general framework, we determined the complexity of {\sc Weighted Subset Odd Cycle Transversal} and {\sc Weighted Subset Feedback Vertex Set} on $H$-free graphs except when $H\in \{2P_1+P_3$, $P_1+P_4$, $2P_1+P_4\}$.
In particular, our results demonstrate that the classifications of {\sc Weighed Subset Odd Cycle Transversal} and {\sc Subset Odd Cycle Transversal} do not coincide for $H$-free graphs.

We believe that the case $H=2P_1+P_3$ is polynomial-time solvable for both problems using the methodology of our framework and our algorithms for $H=P_1+P_3$ as a subroutine. We leave this for future research.
The other two cases are open  even for {\sc Odd Cycle Transversal} and {\sc Feedback Vertex Set}. For these cases we first need to be able to determine the complexity of 
finding a maximum induced disjoint union of stars in a $(P_1+P_4)$-free graph.
We refer to Table~\ref{t-table} for other unresolved cases  in our framework and note again that our results demonstrate that the classifications of {\sc Weighed Subset Odd Cycle Transversal} and {\sc Subset Odd Cycle Transversal} do not coincide for $H$-free graphs.

We note finally that that there are other similar transversal problems that have been studied, but their complexity classifications on $H$-free graphs have not been settled:  {\sc (Subset) Even Cycle Transversal}~\cite{KKK12,MRRS12,PPR21}, for example.  Versions of the transversal problems that we have considered that have the additional constraint that the transversal must induce either a connected graph or an independent set have also been studied for $H$-free graphs~\cite{BDFJP19,CHJMP18,DJPPZ18,JPP20}.  An interesting direction for further research is to consider the subset variant of these problems, and, more generally, to understand the relationships amongst the computational complexities of all these problems.


\appendix

\section{The Proof of Lemma~\ref{l-sp2}}\label{a-sp2}

We first need the following two well-known results: the first is due to Balas and Yu and the second  to Tsukiyama, Ide, Ariyoshi, and Shirakawa.

\begin{lemma}[\cite{BY89}]\label{t-by}
For every constant~$s\geq 1$, the number of maximal independent sets of an $sP_2$-free graph on $n$ vertices is at most $n^{2s}+1$.
\end{lemma}

\begin{lemma}[\cite{TIAS77}]\label{t-tias}
For every constant~$s\geq 1$, it is possible to enumerate all maximal independent sets of an $sP_2$-free graph $G$ on $n$ vertices and $m$ edges
with a delay of ${\mathcal O}(nm)$.
\end{lemma}

\noindent
We are now ready to show Lemma~\ref{l-sp2}.

\medskip
\noindent
{\bf Lemma~\ref{l-sp2} (restated).}
{\it For every integer $s\geq 1$, {\sc Weighted Odd Cycle Transversal} is polynomial-time solvable for $sP_2$-free graphs.}

\begin{proof}
Let $S$ be a minimum-weight odd cycle transversal of an $sP_2$-free graph $G=(V,E)$. Let $B_S=V\setminus S$. Then $G[B_S]$ is a bipartite graph. We choose a bipartition $(X,Y)$ of $B_S$ such that $X$ has maximum size (so every vertex in $Y$ has at least one neighbour in $X$). Then $X$ is a maximal independent set of $G$, as otherwise there exists a vertex~$u\in S$ not adjacent to any vertex of $X$, and thus $S\setminus \{u\}$ is an odd cycle transversal of $G$ with smaller weight, contradicting the fact that $S$ has minimum weight. Moreover, by a similar argument, $Y$ is a maximal independent set of~$G-X$. 
We conclude that $S$ is a minimal odd cycle transversal of $G$.

We describe a procedure to find all minimal odd cycle transversals of $G$. As shown above, the minimum-weight odd cycle transversals of $G$ will be amongst them. Hence, this provides an algorithm for {\sc Weighted Odd Cycle Transversal}.  We enumerate all maximal independent sets of~$G$, and for each maximal independent set~$X$, we enumerate all maximal independent sets of~$G-X$.
For each such set $Y$, we note that $V(G)\setminus (X\cup Y)$ is an odd cycle transversal of $G$.  By the arguments above, we will find every minimal odd cycle transversal in this way, and, by Lemmas~\ref{t-by} and~\ref{t-tias}, this takes polynomial time. \qed
\end{proof}

\section{Weighted Odd Cycle Transversal on ${\mathbf{(sP_1+P_3)}}$-free Graphs}\label{a-b}

We need two known results. For the second one, we use a folklore trick to extend the result of~\cite{GKPP19}
 from $P_6$-free graphs to $(sP_1+P_6)$-free graphs (see also~\cite{DFJPPP19}).
 Let $c$ be the function on the non-negative integers given by $c(s):=\max\{3,2s-1\}$.
We will use this function~$c$ throughout the remainder of this section, starting with the following lemma.

\begin{lemma}[\cite{DFJPPP19}]\label{sp1p3-s}
Let $s \geq 0$ be an integer.
Let~$G$ be a bipartite $(sP_1+\nobreak P_3)$-free graph.
If~$G$ has a connected component on at least~$c(s)$ vertices, then there are at most~$s-1$ other connected components of~$G$ and each of them is on at most two vertices.
\end{lemma}

\begin{theorem}[\cite{GKPP19}]\label{sp1p6-vc}
For every $s\geq 0$, {\sc Weighted Vertex Cover} can be solved in polynomial time on $(sP_1+\nobreak P_6)$-free graphs.
\end{theorem}

\noindent
We can now prove the result of this section by mimicking the proof of the unweighted variant 
from~\cite{DFJPPP19}.

\begin{theorem}\label{sp1p3w-octi}
For every $s\geq 0$, {\sc Weighted Odd Cycle Transversal} can be solved in polynomial time on $(sP_1+\nobreak P_3)$-free graphs.
\end{theorem}

\begin{proof}
Let $s \geq 0$ be an integer, and let $G=(V,E)$ be an $(sP_1+\nobreak P_3)$-free graph with a vertex weighting $w$.
We must describe how to find a minimum weight odd cycle transversal of~$G$.
If $s=0$, then we can use Theorem~\ref{tt-p4}.
We now assume that $s\geq 1$ and use induction.
We will in fact describe how to solve the complementary problem and find a maximum weight induced bipartite subgraph of~$G$.
The proof is by induction on~$s$ and our algorithm performs two steps in polynomial time, which together cover all possibilities.

\thmstep{\label{step4:1}Compute a maximum weight induced bipartite subgraph~$B$ such that every connected component of~$B$ has at least~$c(s)$ vertices.}
By Lemma~\ref{sp1p3-s}, we know that~$B$ will be connected. Hence, $B$ has a unique bipartition, which we denote~$\{X,Y\}$.
We first find a maximum weight induced bipartite subgraph~$B$ that is a star: we consider each vertex~$x$ and find a maximum weight induced star centred at~$x$ by finding a maximum weight independent set in~$N(x)$.
This can be done in polynomial time by Theorem~\ref{sp1p6-vc}.

Next, we find a maximum weight induced bipartite subgraph~$B$ that is not a star.
We consider each of the~$O(n^2)$ choices of edges~$xy$ of~$G$ and find a maximum weight induced connected bipartite subgraph~$B$ such that $x \in X$ and $y \in Y$ and neither~$x$ nor~$y$ has degree~$1$ in~$B$ (since~$B$ is not a star, it must contain such a pair of vertices).
Note that the number of vertices in~$X$ non-adjacent to~$y$ is at most $s-1$, otherwise~$B$ induces an $sP_1+\nobreak P_3$.
Similarly there are at most $s-1$ vertices in~$Y$ non-adjacent to~$x$.
We consider each of the~$O(n^{2s-2})$ possible pairs of disjoint sets~$X'$ and~$Y'$, which are each independent sets of size at most $s-1$ such that $X'\cup Y'$ is anti-complete to~$\{x,y\}$.
We will find a maximum weight induced bipartite subgraph with partition classes~$X$ and~$Y$ such that $\{x\} \cup X' \subseteq X$ and $\{y\} \cup Y' \subseteq Y$ and every vertex in $X \setminus X'$ is adjacent to~$y$ and every vertex in $Y \setminus Y'$ is adjacent to~$x$.
That is, we must find a maximum weight independent set in both $N(x) \setminus N(\{y\} \cup Y')$ and $N(y) \setminus N(\{x\} \cup X')$; see \figurename~\ref{f-il} for an illustration.
This can be done in polynomial time, again by applying Theorem~\ref{sp1p6-vc}.

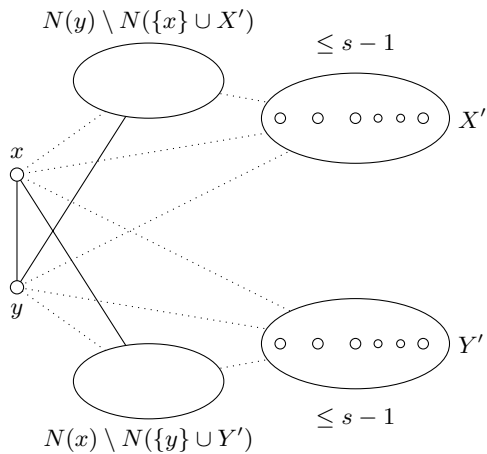
\begin{figure}
\begin{center}
\begin{tikzpicture}[xscale=0.5, yscale=0.5]
\draw (-1,-4) -- (-4.5,1.5) -- (-4.5,-1.5) -- (-1,4);
\draw[dotted] (4.5,3) -- (-4.5,1.5) -- (4.5,-3) (-4.5,1.5) -- (-1,4) -- (4.5,3)
(4.5,3) -- (-4.5,-1.5) -- (4.5,-3) (-4.5,-1.5) -- (-1,-4) -- (4.5,-3);
\draw[fill=white] 
(-4.5,1.5) circle [radius=5pt]
(-4.5,-1.5) circle [radius=5pt]
(2.5,3) circle [radius=4pt]
(3.5,3) circle [radius=4pt]
(4.5,3) circle [radius=4pt]
(5.1,3) circle [radius=3pt]
(5.7,3) circle [radius=3pt]
(6.3,3) circle [radius=4pt]
(2.5,-3) circle [radius=4pt]
(3.5,-3) circle [radius=4pt]
(4.5,-3) circle [radius=4pt]
(5.1,-3) circle [radius=3pt]
(5.7,-3) circle [radius=3pt]
(6.3,-3) circle [radius=4pt]
(-1,4) ellipse (2cm and 1cm)
(-1,-4) ellipse (2cm and 1cm)
(4.5,3) ellipse (2.5cm and 1.2cm)
(4.5,-3) ellipse (2.5cm and 1.2cm);
\node[above] at (4.5,4.5) {$\leq s-1$};
\node[below] at (4.5,-4.5) {$\leq s-1$};
\node[above] at (-1,5) {$N(y)\setminus N(\{x\}\cup X')$};
\node[below] at (-1,-5) {$N(x)\setminus N(\{y\}\cup Y')$};
\node[right] at (7,3) {$X'$};
\node[right] at (7,-3) {$Y'$};
\node[above] at (-4.5,1.7) {$x$};
\node[below] at (-4.5,-1.7) {$y$};
\end{tikzpicture}
\end{center}
\caption{An illustration~\cite{DFJPPP19} of Step~\ref{step4:1} of the algorithm in the proof of Theorem~\ref{sp1p3w-octi}. Full and dotted lines indicate when two sets are complete or anti-complete to each other, respectively. The absence of a full or dotted lines indicates that edges may or may not exist between two sets.}
\label{f-il}
\end{figure}

\thmstep{\label{step4:2}Compute a maximum weight induced bipartite subgraph~$B$ such that~$B$ has a connected component with at most $c(s)-1$ vertices.}
We consider each of the~$O(n^{c(s)-1})$ possible choices of a non-empty set~$L$ of at most $c(s)-1$ vertices and discard those that do not induce a bipartite graph.
We will find a maximum weight induced bipartite subgraph~$B$ that has~$G[L]$ as a connected component.
Let $U=N(L)$, and let $G'=G-(L\cup U)$.
As~$G'$ is $((s-\nobreak 1)P_1+\nobreak P_3)$-free, we can find a maximum weight induced bipartite subgraph~$B'$ of~$G'$ in polynomial time and $B' + G[L]$ is a maximum weight induced bipartite subgraph among those that have~$G[L]$ as a connected component.
\qed
\end{proof}

\section{The Proof of Theorem~\ref{tt-p4}}\label{a-p4}

In this section we give a direct proof of Theorem~\ref{tt-p4}, which states that \textsc{Weighted Subset Odd Cycle Transversal} can be solved in polynomial-time for $P_4$-free graphs.  Our approach mimics the proof of the unweighted variant of  {\sc Subset Odd Cycle Transversal} from~\cite{BJPP20}. 

Let~$G_1$ and~$G_2$ be two vertex-disjoint graphs.
The \emph{union} operation $+$ creates the disjoint union $G_1+\nobreak G_2$ of~$G_1$ and~$G_2$ (recall that $G_1+G_2$ is the graph with vertex set $V(G_1)\cup V(G_2)$ and edge set $E(G_1)\cup E(G_2)$). 
The \emph{join} operation adds an edge between every vertex of~$G_1$ and every vertex of~$G_2$.
The graph~$G$ is a \emph{cograph} if~$G$ can be generated from independent vertices 
 by a sequence of join and union operations.
A graph is a cograph if and only if it is $P_4$-free (see, for example,~\cite{BLS99}).
It is also well known~\cite{CLS81} that a graph $G$ is a cograph if and only if $G$ allows a unique tree decomposition called the {\it cotree} $T_G$ of $G$, which has the following properties:
\begin{itemize}
\item [1.]  The root $r$ of $T_G$ corresponds to the graph $G_r=G$.
\item [2.]   Each leaf $x$ of $T_G$ corresponds to exactly one vertex of $G$, and vice versa. Hence $x$ corresponds to a unique single-vertex graph $G_x$.
\item [3.]  Each internal node $x$ of $T_G$ has at least two children, is  labelled $\oplus$ or $\otimes$, and corresponds to an induced subgraph $G_x$ of $G$ defined as follows:
\begin{itemize}
\item if $x$ is a $\oplus$-node, then $G_x$ is the disjoint union of all graphs $G_y$ where $y$ is a child of~$x$;
\item if $x$ is a $\otimes$-node, then $G_x$ is the join of all graphs $G_y$ where $y$ is a child of~$x$.
\end{itemize}
\item [4.] Labels of internal nodes on the (unique) path from any leaf to $r$ alternate between $\oplus$ and~$\otimes$.
\end{itemize}
Note that $T_G$ has ${\mathcal O}(n)$ vertices. We modify $T_G$ into a {\it modified cotree} $T_G'$ in which each internal node has exactly two children but (4) no longer holds.  The following well-known procedure (see for example~\cite{BM93}) achieves this. If an internal node~$x$ of $T_G$ has more than two children $y_1$ and $y_2$, remove the edges $xy_1$ and $xy_2$ and add a new vertex $x'$ with 
edges $xx'$, $x'y_1$ and $x'y_2$. If~$x$ is a $\oplus$-node, then $x'$ is a $\oplus$-node. If $x$ is a $\otimes$-node, then $x'$ is a 
$\otimes$-node. Applying this rule exhaustively yields $T_G'$. As~$T_G$ has ${\mathcal O}(n)$ vertices, constructing $T_G'$ from $T_G$ takes linear time. This leads to the following result, due to Corneil, Perl and Stewart, who proved it for cotrees.

\begin{lemma}[\cite{CPS85}]\label{l-cotree}
Let $G$ be a graph with $n$ vertices and $m$ edges.  Then deciding whether or not $G$ is a cograph, and constructing a modified cotree $T_G'$ (if it exists) takes time ${\mathcal O}(n+m)$.
\end{lemma}

We require a result for the weighted subset variant of {\sc Vertex Cover}, which we now formally define.
For a graph $G=(V,E)$ and a set $T \subseteq V$, 
a set $S_T\subseteq V$~is a {\it $T$-vertex cover} if $S_T$ has at least contains one vertex incident to every edge that is incident to a vertex of $T$.

\problemdef{{\sc Weighted Subset Vertex Cover}}{a graph $G$, a subset $T\subseteq V(G)$, a non-negative vertex weighting~$w$ and an integer $k\geq 1$.}{does $G$ have a $T$-vertex cover $S_T$ with $w(S_T) \leq k$?} 
The following result is proven in the same way as  the unweighted variant of {\sc Subset Vertex Cover} in \cite{BJPP20}.

\begin{lemma}\label{svc-p4}
{\sc Weighted Subset Vertex Cover} can be solved in polynomial time for $P_4$-free graphs.
\end{lemma}

\begin{proof}
Let $G$ be a cograph with $n$ vertices and $m$ edges. First construct a modified cotree~$T_G'$ and then consider each node of $T_G'$  starting at the leaves of $T_G'$ and ending at the root $r$. Let $x$ be a node of $T_G'$. We let $S_x$ denote a minimum 
weight $(T\cap V(G_x))$-vertex cover of $G_x$. 

If $x$ is a leaf, then 
$G_x$ is a $1$-vertex graph. Hence, we can let $S_x=\emptyset$.
Now suppose that~$x$ is a $\oplus$-node.
Let $y$ and $z$ be the two children of $x$. Then, as $G_x$ is the disjoint union of $G_y$ and $G_z$, we can let $S_x=S_y\cup S_z$.
Finally suppose that  $x$ is a $\otimes$-node.
Let $y$ and $z$ be the two children of $x$. As $G_x$ is the join of $G_y$ and $G_z$ we observe the following:
if $V(G_x)\setminus S_x$ contains a vertex of $T\cap V(G_y)$, then $V(G_z)\subseteq S_x$. 
Similarly, if $V(G_x)\setminus S_x$ contains a vertex of $T\cap V(G_z)$, then $V(G_y)\subseteq S_x$. 
Hence, we  let $S_x$ be a set with minimum weight from $S_y\cup V(G_z)$, $S_z\cup V(G_y)$ and 
$T\cap V(G_x)$.

Constructing $T_G'$ takes ${\mathcal O}(n+m)$ time by Lemma~\ref{l-cotree}.
As $T_G'$ has~${\mathcal O}(n)$ nodes and processing a node takes ${\mathcal O}(1)$ time, the total running time is ${\mathcal O}(n+m)$. \qed 
\end{proof}

\noindent
We are now ready to prove the following result.

\medskip
\noindent
\faketheorem{Theorem~\ref{tt-p4} (restated).}
{\it {\sc Weighted Subset Odd Cycle Transversal} is polynomial-time solvable for $P_4$-free graphs.}

\begin{proof}
Let $G$ be a cograph with $n$ vertices and $m$ edges with an edge weighting~$w$. Let $T\subseteq V(G)$. First construct the modified cotree~$T_G'$ and then consider each node of $T_G'$ starting at the leaves of $T_G'$ and ending in its root $r$. Let~$x$ be a node of $T_G'$. We let $S_x$ denote a minimum weight odd $(T\cap V(G_x))$-cycle transversal of $G_x$. 

If $x$ is a leaf, then $G_x$ is a 1-vertex graph. Hence, we can let $S_x=\emptyset$.
Now suppose that~$x$ is a $\oplus$-node.
Let $y$ and $z$ be the two children of $x$. Then, as $G_x$ is the disjoint union of $G_y$ and $G_z$, we let $S_x=S_y\cup S_z$.

Finally suppose that $x$ is a $\otimes$-node.
Let $y$ and $z$ be the two children of $x$. Let $T_y=T\cap V(G_y)$ and $T_z=T\cap V(G_z)$.
Let $B_x=V(G_x)\setminus S_x$. 
As $G_x$ is the join of $G_y$ and $G_z$ we observe the following.
If $B_x\cap V(G_y)$ contains two adjacent vertices, at least one of which belongs to~$T_x$, then $B_x\cap V(G_z)=\emptyset$ (as otherwise $G[B_x]$ has a triangle containing a vertex of $T$) and thus 
$V(G_z)\subseteq S_x$. 
In this case we may assume that $S_x=S_y\cup V(G_z)$.
Similarly, if $B_x\cap V(G_z)$ contains two adjacent vertices, at least one of which belongs to $T_z$, then $B_x\cap V(G_y)=\emptyset$ and thus $V(G_y)\subseteq S_x$. 
In this case we may assume that $S_x=S_z\cup V(G_y)$.
From the two sets $S_y\cup V(G_z)$ and $S_z\cup V(G_y)$ we remember the one that has smallest weight.
 
It remains to examine the case where $B_x\cap V(G_y)$ and $B_x\cap V(G_z)$ induce subgraphs of~$G$ in which the vertices of $T_y\cap B_x$ and $T_z\cap B_x$, respectively, are singleton components.

First suppose that  $T_y\cap B_x$ and $T_z\cap B_x$ are both non-empty.
Then $B_x\cap V(G_y)$ and $B_x\cap V(G_z)$ are both independent sets, as otherwise $G[B_x]$ would contain a $T$-triangle.
We examine this situation by computing a maximum weight independent set $I_y$ in $G_y$ and a maximum weight independent set $I_z$ in $G_z$; it is well-known that this can be done in polynomial time (for example, it follows from Lemma~\ref{svc-p4}).
We remember $V(G_x)\setminus (I_y\cup I_z)$.

Now suppose that  $T_y\cap B_x$ is non-empty, but $T_z\cap B_x$ is empty.
Then $B_x\cap V(G_z)$ must be an independent set, as otherwise we obtain a $T$-triangle by taking a vertex of $T_y\cap B_x$ and two adjacent vertices of $B_x\cap V(G_z)$.
First assume that $B_x\cap V(G_z)$ has size at least~$2$.
We observe that $(B_x\cap V(G_y))\setminus T_y$ is also an independent set; otherwise two adjacent vertices of $(B_x\cap V(G_y))\setminus T_y$, two vertices of $B_x\cap V(G_z)$ and one vertex of $T_y\cap B_x$ would form a $T$-cycle on five vertices. Hence, both $B_x\cap V(G_y)$ and $B_z\cap V(G_z)$ are independent sets, and we already dealt with this case above.

Now assume that $B_x\cap V(G_z)$ has size at most~$1$. 
In this case $B_x\cap V(G_y)$ is a minimum weight $T_y$-vertex cover of $G_y$. 
We can compute a minimum weight $T_y$-vertex cover $S$ of $G_y$ in polynomial time by Lemma~\ref{svc-p4}. 
We remember 
$S\cup (V(G_z)\setminus \{z\})$
where $z$ is a vertex with maximum weight in $V(G_z)\setminus T_z$ if the latter set is non-empty; otherwise we just remember
$S\cup (V(G_z)$.

We deal with the case where $T_z\cap B_x$ is non-empty, but $T_y\cap B_x$ is empty in the same way and remember the output. 
We also consider the possible situation where $T_z\cap B_x=T_y\cap B_x=\emptyset$, in which case 
we remember $T$.
Finally, we take as set $S_x$ a set of minimum size over the sets that we remembered.

Constructing $T_G'$ takes ${\mathcal O}(n+m)$ time by Lemma~\ref{l-cotree}.
As~$T_{G'}$ has~$\mathcal{O}(n)$ nodes and processing a node takes $O(n+m)$ time (due to the application of Lemma~\ref{svc-p4}), the total running time is $\mathcal{O}(n^2+mn)$.
\qed
\end{proof}

\section{The Proof of Theorem~\ref{tt-5p1}}\label{a-5p1}

The result is an analogue of Papadopoulos and Tzimas's hardness result for {\sc Weighted Subset Feedback Vertex Set} on $5P_1$-free graphs \cite[Theorem~2]{PT20}.  As mentioned, the proof is essentially identical, as all the relevant $T$-cycles in the constructed {\sc Weighted Subset Feedback Vertex Set} instance are odd.  We provide the full proof for completeness.

A {\it vertex cover} of a graph $G=(V,E)$ is a set $S\subseteq V$ such that $G-S$ is an independent set.
The corresponding decision problem is defined as follows:

\problemdef{{\sc Vertex Cover}}{a graph $G$ and an integer $k\geq 1$.}{does $G$ have a vertex cover $U$ with $|U|\leq k$?} 

\medskip
\noindent
\faketheorem{Theorem~\ref{tt-5p1} (restated).}
{\it {\sc Weighted Subset Odd Cycle Transversal} is \NP-complete for $5P_1$-free graphs.
}

\begin{proof}
We reduce from the {\sc Vertex Cover} problem on $3$-partite graphs, which is \NP-complete~\cite{GJ79}.
Let $(G, k)$ be a {\sc Vertex Cover} instance where $G$ is a $3$-partite graph such that $(X_1,X_2,X_3)$ is a partition of $V(G)$ into independent sets.
We construct an instance $(G', T, w, k)$ of {\sc Weighted Subset Odd Cycle Transversal} as follows.
First, let $G'$ be the graph obtained from $G$ by making each of $X_1$, $X_2$, and $X_3$ into cliques, then introducing $4$ new vertices $r_1$, $r_2$, $r_3$, and $t$ where the neighbourhood of $r_i$ is $X_i \cup \{t\}$ for each $i \in \{1,2,3\}$, and the neighbourhood of $t$ is $\{r_1,r_2,r_3\}$.
We define the weighting~$w$ on $V(G')$ as follows: let $w(v) = 1$ for each $v \in V(G)$, let $w(r_i) = |V(G)|$ for each $i \in \{1,2,3\}$, and $w(t) = |V(G)|$.
Finally, set $T = \{t\}$.
Observe that $X_i \cup \{r_i\}$ is a clique for each $i \in \{1,2,3\}$.
Thus, $G'$ is $5P_1$-free: an independent set of $G'$ contains at most one vertex from $X_i \cup \{r_i\}$ for each $i \in \{1,2,3\}$, and the only other vertex not in one of these three sets is $t$.
See Figure~\ref{f-5p1} for an illustration.

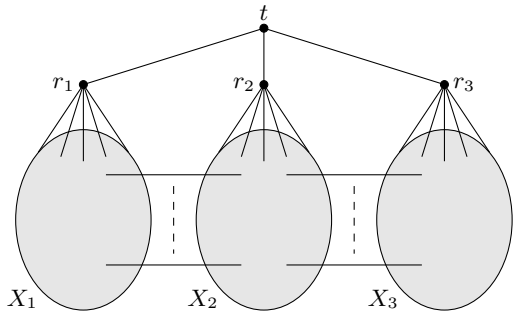
\begin{figure}
\begin{center}
\begin{tikzpicture}[scale=0.3]
\draw[color=black, fill=gray!20] (0,-2) ellipse (3cm and 4cm);
\draw[color=black, fill=gray!20] (8,-2) ellipse (3cm and 4cm);
\draw[color=black, fill=gray!20] (16,-2) ellipse (3cm and 4cm);
\filldraw [black] (0,4) circle [radius=5pt] (8,4) circle [radius=5pt] (16,4) circle [radius=5pt] (8,6.5)  circle [radius=5pt];
\draw (0,4)--(8,6.5)--(8,4);
\draw(8,6.5)--(16,4);
\draw (0,4)--(-2,1);
\draw (0,4)--(-1,0.8);
\draw (0,4)--(0,0.6);
\draw (0,4)--(1,0.8);
\draw (0,4)--(2,1);
\draw (8,4)--(6,1);
\draw (8,4)--(7,0.8);
\draw (8,4)--(8,0.6);
\draw (8,4)--(9,0.8);
\draw (8,4)--(10,1);
\draw (16,4)--(14,1);
\draw (16,4)--(15,0.8);
\draw (16,4)--(16,0.6);
\draw (16,4)--(17,0.8);
\draw (16,4)--(18,1);

\draw (1,0)--(7,0);
\draw[dashed] (4,-0.5)--(4,-3.5);
\draw (1,-4)--(7,-4);
\draw (9,0)--(15,0);
\draw[dashed] (12,-0.5)--(12,-3.5);
\draw (9,-4)--(15,-4);

\node[left] at (0,4) {$r_1$};
\node[left] at (8,4) {$r_2$};
\node[right] at (16,4) {$r_3$};
\node[above] at (8,6.5) {$t$};
\node at (-2.7,-5.5) {$X_1$};
\node at (5.3,-5.5) {$X_2$};
\node at (13.3,-5.5) {$X_3$};
\end{tikzpicture}
\end{center}
\caption{Let $G$ be a 3-partite graph with vertex partition $(X_1,X_2,X_3)$; the illustration is schematic --- any edges might be present between each pair of sets in the partition including $X_1$ and $X_3$ (not indicated for simplicity).  A graph $G'$ is constructed by adding edges so that each of $X_1$, $X_2$ and $X_3$ induces a clique and also adding four further vertices as illustrated; for $1 \leq i \leq 3$, $r_i$ is adjacent to every vertex in $X_i$.  In the proof of the theorem, we use that a subset of $V(G)$ is a vertex cover of $G$ if and only if it is an odd $T$-cycle transversal of $G'$ for $T=\{t\}$.}
\label{f-5p1}
\end{figure}

Suppose that $G$ has a vertex cover $U$ of size at most $k$.
We claim that $U$ is an odd $T$-cycle transversal of $G'$ with $w(U) \le k$.
Clearly $w(U) \le k$, since $w(u) = 1$ for each $u \in U \subseteq V(G)$.
It remains to show that $U$ is an odd $T$-cycle transversal.
Towards a contradiction, suppose that $G'-U$ contains an odd $T$-cycle~$C$.  Since $T=\{t\}$, the cycle $C$ contains an edge incident to $t$.  Without loss of generality, we may assume that $C$ contains the edge $tr_1$.
Observe that $V(G) \setminus U$ is an independent set of $G$, so each edge of $G' - U$ is either incident to (at least) one of $r_1$, $r_2$, $r_3$ and $t$, or both endpoints are in $X_i$ for some $i \in \{1,2,3\}$.
Now every path in $G'-U$ from $v_1$ to $v_2$ for $v_1 \in X_1 \cup \{r_1\}$ and $v_2 \in V(G) \setminus (X_1 \cup \{r_1\})$ passes through the vertex $t$.
So $tr_1$ is a bridge, implying it is not contained in a cycle, a contradiction.
We deduce that $U$ is an odd $T$-cycle transversal.

Now suppose that $G'$ has an odd $T$-cycle transversal of weight at most $k$.
Let	 $U$ be a minimum-weight odd $T$-cycle transversal of $G'$.  In particular, $w(U) \le k$.
Observe that for any $v \in X_1 \cup X_2 \cup X_3$, the set $S_0=(X_1 \cup X_2 \cup X_3) \setminus \{v\}$ is an odd $T$-cycle transversal, since $G'[\{r_1,r_2,r_3,t,v\}]$ is a tree. As $S_0$ is an odd $T$-cycle transversal and $U$ is a minimum-weight odd $T$-cycle transversal, $w(U)\leq w(S_0)=|V(G)|-1$. Hence, $U$ does not contain $r_1$, $r_2$, $r_3$, or $t$.
That is, $U \subseteq V(G)$ and thus $w(U)=|U|$. As $w(U)\le k$, this implies that $|U| \le k$.

We claim that $U$ is a vertex cover of $G$.
Suppose not.
Then, without loss of generality, there is an edge $x_1x_2 \in E(G)$ such that $x_1 \in X_1$, $x_2 \in X_2$, and $x_1,x_2 \notin U$.  But then $tr_1x_1x_2r_2t$ is an odd $T$-cycle of $G'-U$, a contradiction.
We deduce that $U$ is a vertex cover of $G$ of size at most $k$, which completes the proof. \qed
\end{proof}

\end{document}